\documentclass[lettersize,journal]{IEEEtran}
\usepackage{amsmath,amsfonts}
\usepackage{algorithmic}
\usepackage{algorithm}
\usepackage{array}
\usepackage[caption=false,font=footnotesize,labelfont=rm,textfont=rm]{subfig}
\usepackage{textcomp}
\usepackage{stfloats}
\usepackage{url}
\usepackage{verbatim}
\usepackage{graphicx}
\usepackage{cite}
\usepackage{hyperref}
\usepackage{tabularray}
\usepackage{color}
\usepackage{booktabs}
\usepackage{amsmath}
\usepackage{amssymb}
\usepackage{tikz}
\usetikzlibrary{shapes.geometric, arrows}

\tikzstyle{startstop} = [rectangle, rounded corners, minimum width=1.2cm, minimum height=1cm,text centered, draw=black, fill=red!30]
\tikzstyle{io} = [trapezium, trapezium left angle=70, trapezium right angle=110, minimum width=1cm, minimum height=1cm, text centered, draw=black, fill=blue!30]
\tikzstyle{process} = [rectangle, minimum width=1cm, minimum height=1cm, text centered, draw=black, fill=orange!30]
\tikzstyle{decision} = [diamond, minimum width=1cm, minimum height=1cm, text centered, draw=black, fill=green!30]
\tikzstyle{arrow} = [thick,->,>=stealth]

\newtheorem{theorem}{Theorem} 
\newtheorem{definition}{Definition} 
\newtheorem{lemma}{Lemma} 

\begin{document}

\title{3D Spherical Directly-Connected Antenna Array for Low-Altitude UAV Swarm ISAC}

\author{Haoyu Jiang, Zhenjun Dong, Zhiwen Zhou, Yong Zeng,~\IEEEmembership{Fellow,~IEEE}

\thanks{ 
This work was supported by the National Natural Science Foundation of China under Grant 62571116, and Natural Science Foundation for Distinguished Young Scholars of Jiangsu Province with grant number BK20240070.





H. Jiang, Z. Zhou, and Y. Zeng are with the National Mobile Communications Research Laboratory, Southeast University, Nanjing 210096, China. Y. Zeng is also with the Purple Mountain Laboratories, Nanjing 211111, China. Z. Dong is now with the Purple Mountain Laboratories, Nanjing 211111, China. She was with the National Mobile Communications Research Laboratory, Southeast University, Nanjing 210096, China. (e-mail: \{230258936, zhenjun\_dong, zhiwen\_zhou, yong\_zeng\}@seu.edu.cn). (Corresponding author: Yong Zeng.)

}

}



\maketitle

\begin{abstract}
Recently a novel multi-antenna architecture termed ray antenna array (RAA) was proposed, where several simple uniform linear arrays (sULAs) are arranged in a ray-like structure to enhance communication and sensing performance. By eliminating the need for phase shifters, it also significantly reduces hardware costs. However, RAA is prone to signal blockage and has no elevation angle resolution capability in three-dimensional (3D) scenarios. To address such issues, in this paper we propose a novel spherical directly-connected antenna array (DCAA), which composes of multiple simple uniform planar arrays (sUPAs) placed over a spherical surface. All elements within each sUPA are directly connected. Compared to conventional arrays with hybrid analog/digital beamforming (HBF), DCAA significantly reduces hardware cost, improves energy focusing, and provides superior and uniform angular resolution for 3D space. These advantages make DCAA particularly suitable for integrated sensing and communication (ISAC) in low-altitude unmanned aerial vehicles (UAV) swarm scenarios, where targets may frequently move away from the boresight of traditional arrays, degrading both communication and sensing performance. 
Simulation results demonstrate that the proposed spherical DCAA achieves significantly better angular resolution and higher spectral efficiency than conventional array with  HBF, highlighting its strong potential for UAV swarm ISAC systems.
\end{abstract}


\section{Introduction}
\IEEEPARstart{L}{ow-altitude} unmanned aerial vehicle (UAV) swarms, characterized by their massive number, collective intelligence, and collaborative operation, are expected to play an important role in future low-altitude economy\cite{Yong_UAV,yuxuan,Low_Altitude_Economy1,Low_Altitude_Economy2,UAV_communication}. By overcoming the inherent limitations of single UAV in terms of coverage, payload, and endurance, swarms can significantly expand the application horizons of aerial platforms, enabling advanced services such as large-scale coordinated delivery, environmental monitoring, and smart city infrastructure management\cite{UAV_swarm0,UAV_swarm1,UAV_swarm4}. However, the expected rapid increase in aerial traffic density and operational complexity within low-altitude airspace creates stringent requirements for both reliable communication and accurate sensing\cite{UAV_swarm2,UAV_swarm3}. This necessity has sparked significant interest in integrated sensing and communication (ISAC), a key enabling technology for 6G networks that fuses sensing and communication seamlessly using shared spectral resources, hardware, and signal processing frameworks\cite{liufan_ISAC,Qianglong_tutorial,xiaozhiqiang_ISAC,xiaoISAC}. For UAV swarm applications, ISAC is essential for ensuring safe and efficient operations, facilitating tasks such as real-time swarm localization, collision avoidance, and cooperative trajectory planning\cite{UAV_swarm00,UAV3}.

Despite its promise, deploying ISAC for low-altitude UAV swarms introduces unique challenges. UAVs are inherently constrained by size, weight, and power limitations, which restrict their onboard communication and sensing capabilities\cite{Yong_UAV,UAV,UAV2}. Furthermore, the high mobility and 3D spatial distribution of swarm members require ISAC systems with exceptional angular resolution to distinguish closely spaced targets and support high-gain beamforming for communication\cite{ISAC_UAV,UAV_survey}. Conventional multiple-input multiple-output (MIMO) architectures, which form the backbone of most contemporary ISAC systems, often struggle to meet these demands cost-effectively\cite{haiquan_XLMIMO,SparseMIMO}. For example, traditional array architectures like uniform linear arrays (ULAs) and uniform planar arrays (UPAs) suffer from fundamental limitations: their angular resolution degrades significantly for signal directions away from the array boresight, and achieving high resolution requires a large physical aperture or a massive number of antenna elements\cite{RAA_ISAC}. More critically, realizing such arrays with fully digital or hybrid analog/digital beamforming (HBF) require a large number of radio frequency (RF) chains and phase shifters, which are expensive and power-hungry especially at high-frequency bands like millimeter-wave (mmWave) and Terahertz (THz)\cite{Analog_Beamforming,Hybrid_Precoding}. This creates a critical cost-performance trade-off that hinders the scalable deployment of ISAC for widespread UAV swarm applications.

To address the hardware cost and implementation challenges in high-frequency systems, researchers have explored various innovative antenna architectures that reduce the required phase shifters or RF chains. Notable examples include lens antenna arrays\cite{lens_antenna_array,7416205}, which transform the signal from the antenna space to a lower-dimensional beamspace using electromagnetic lenses, thereby significantly reducing the number of required RF chains. However, they often require bulky and precise lens components, which increase size and cost. Fluid antennas\cite{Fluid_Antenna,Fluid_Antenna2} and movable antennas\cite{Movable_Antennas,Movable_Antennas2,Movable_Antenna_dong} represent another direction, where the antenna shape or position can be dynamically optimized at the transmitter or receiver to enhance performance with fewer physical elements. While these systems offer flexibility and improved spectral efficiency, they typically involve mechanical or reconfigurable parts that introduce additional complexity, power consumption, and response latency. Another approach is the pinching antenna\cite{Pinching_Antenna,Pinching_Antenna2}, which uses small dielectric particles inserted into waveguides to create line-of-sight (LoS) paths for users, though it may face challenges in scalability and integration. Furthermore, the tri-hybrid MIMO architecture\cite{Tri_Hybrid_MIMO,Tri_Hybrid_MIMO2} incorporates reconfigurable antennas along with both digital and analog precoding to balance performance and hardware complexity, yet it still relies on reconfigurable elements that add to cost and control overhead.

Recently, ray antenna array (RAA)\cite{RAA_origin} was proposed, which offers a distinctly different and cost-effective solution. RAA comprises multiple simple uniform linear arrays (sULAs), where all antenna elements within each sULA are directly connected without any phase shifters\cite{RAAjournal_dong,RAAzhu,RAA_ISAC}. Each sULA naturally forms a beam whose main lobe direction aligns with its physical orientation. This design eliminates phase shifters, drastically reduces the number of required RF chains through a selection network, and achieves uniform angular resolution independent of the signal angle of arrival (AoA). This is a distinct advantage over conventional linear or planar arrays, whose resolution deteriorates at larger angles. However, the basic RAA proposed in \cite{RAA_origin} is inherently a two-dimensional (2D) structure designed for discriminating only the azimuth or elevation angle. For practical low-altitude UAV swarm ISAC, where targets are distributed in 3D space, the ability to jointly discriminate both azimuth and elevation angles with high resolution is essential. Besides, a practical deployment challenge for the RAA is the potential blockage or mutual coupling between adjacent sULAs in a planar configuration. This can lead to signal attenuation and inter‑ray interference, ultimately degrading system performance. The author in \cite{cylinderRAA} proposed cylinder directly-connected antenna array (DCAA) to address the blockage problem by utilizing multiple simple uniform circular arrays (sUCAs) in a layered 3D configuration, but the 3D coverage issue remains unsolved.

Motivated by the cost and resolution benefits of the RAA principle and to overcome its dimensional and blockage limitation, this paper proposes a novel spherical DCAA for 3D low-altitude UAV swarm ISAC. The spherical DCAA generalizes the RAA concept from a planar to a spherical configuration. It consists of multiple simple uniform planar arrays (sUPAs) distributed over a spherical surface. Crucially, all elements within each individual sUPA are directly connected, requiring no phase shifters within the sUPA. Only a few number of RF chains are shared by all sUPAs via a selection network. This architecture inherits and extends the key advantages of RAA\cite{RAAjournal_dong}: 1) drastically reduced hardware cost and complexity by replacing a massive number of phase shifters with low-cost antenna elements and RF switches; 2) enhanced energy focusing capability, since antenna elements with higher directivity can be used as each sUPA is responsible for a dedicated angular sector; and 3) superior and uniform angular resolution. Beyond these inherited benefits, the spherical DCAA offers two key distinctive advantages. First, it achieves full 3D angular coverage by enabling joint discrimination of both azimuth and elevation angles. Second, the spherical arrangement of sUPAs effectively mitigates mutual blockage and coupling between adjacent sUPAs, thereby ensuring more reliable signal reception across the entire spatial domain.

The main contributions of this paper are summarized as follows:
\begin{itemize}
    \item First, we propose the novel spherical DCAA architecture for 3D low-altitude UAV swarm ISAC. 
    It is composed of multiple sUPAs distributed over a spherical surface, where all the $M \times M$ antenna elements within each sUPA are directly connected without requiring any phase shifter. We establish a complete mathematical signal model for a bistatic ISAC system utilizing this architecture. 
    \item Second, we conduct a rigorous beam pattern analysis for the spherical DCAA and provide a systematic design methodology. Through mathematical derivation, we formally characterize the unique angular resolution properties of spherical DCAA. 
    Based on this analysis, we propose a systematic spherical arrangement of the sUPAs. 
    Furthermore, we address practical implementation constraints by deriving the minimum radius of the spherical structure required to avoid physical collisions between adjacent sUPAs.
    \item Third, we develop an effective sensing parameter estimation algorithm for ISAC within the spherical DCAA framework. A super-resolution multiple signal classification (MUSIC) algorithm is designed to tailor to the structure of the spherical DCAA's equivalent steering vector for jointly estimating the azimuth and elevation angles of arrival.
    \item Finally, we provide extensive simulation results to validate the theoretical analysis and demonstrate the superior performance of the proposed spherical DCAA system under low-altitude UAV swarm scenarios. The performance of the spherical DCAA is benchmarked against conventional UPA with the same array gain that uses the classical Kronecker product codebook (KPC)-based hybrid beamforming. The results show that the spherical DCAA achieves considerable superior angular resolution and significantly better performance in discriminating closely-spaced UAV targets and achieves angle estimation with lower root mean square error (RMSE). Furthermore, the spherical DCAA exhibits enhanced spectral efficiency for communication, thanks to its higher beamforming gain since more directive antenna elements can be used. 
    
\end{itemize}

The remainder of this paper is organized as follows. Section \ref{sec:system model} presents the system model for the spherical DCAA-based bistatic UAV swarm ISAC. Section \ref{sec:DCAA Architecture} introduced the design of spherical DCAA architecture. Section \ref{sec:DCAAvsUPA} details the conventional UPA architecture and compares its beam pattern with spherical DCAA. Section \ref{sec:algorithm} introduces the ISAC signal model and the proposed joint parameter estimation algorithm. Simulation results and discussions are provided in Section \ref{sec:simulation}. Finally, the paper is concluded in Section \ref{sec:conclusion}.

\textit{Notation}: Scalars are denoted by italic letters. Vectors and matrices are denoted by boldface lower-case and upper-case letters, respectively. Transpose and Hermitian transpose are denoted by $(\cdot)^{\mathrm{T}}$ and $(\cdot)^{\mathrm{H}}$, respectively. The absolute value and $l_2$ norm are given by $|\cdot|$  and $\|\cdot\|_2$, respectively. $\mathbb{C}^{m \times n}$ denotes the space of $m \times n$ complex matrices. $\mathrm{card}(\cdot)$ denotes the cardinality of a set. $j = \sqrt{-1}$ denotes the imaginary unit of complex-valued numbers. $\mathbf{1}_N$ is an $N \times 1$ all-ones vector. $\otimes$ denotes the Kronecker product. $\lfloor \cdot \rfloor$ and $\lceil \cdot \rceil$ denote the floor and ceiling functions, respectively. $\delta(\cdot)$ is the Dirac delta function. $\epsilon(\cdot)$ is the Heaviside step function. The distribution of a circularly symmetric complex Gaussian (CSCG) random variable with zero mean and variance $\sigma^2$ is denoted by $\mathcal{CN}(0,\sigma^2)$ and $\sim $ stands for “distributed as”. The set of integers is denoted by $\mathbb{Z}$.

\section{System Model}\label{sec:system model}
\begin{figure*}[htbp]
        \centering \includegraphics[width=0.7\linewidth]{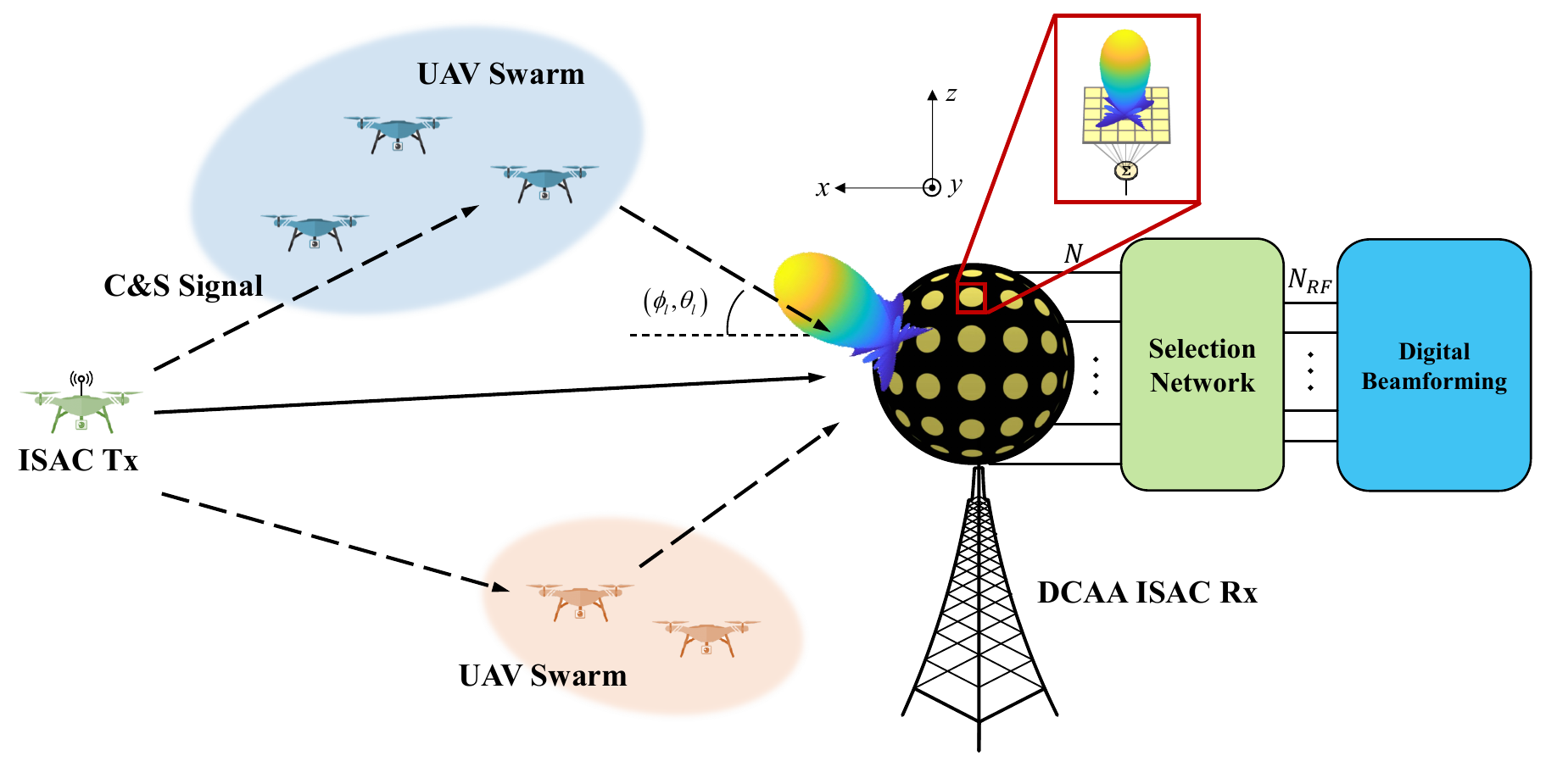}
        \caption{\label{fig:DCAA_env}An illustration of spherical DCAA-based bi-static ISAC for low-altitude UAV swarm.}
\end{figure*}
We consider a low-altitude UAV swarm bi-static ISAC system where the ISAC transmitter has a single antenna and the ISAC receiver (Rx) is equipped with spherical DCAA, as shown in Fig. \ref{fig:DCAA_env}. The system aims to provide communication service for the Tx, while sensing the UAV swarm targets via its transmitted signal. The channel between the Tx and Rx may contain a line-of-sight (LoS) link along with $L$ non-line-of-sight (NLoS) components resulting from reflections by the sensing targets. The $l$th path (where $l = 0$ corresponds to the LoS) is characterized by a set of parameters ${\phi_l, \theta_l}$, denoting the azimuth and elevation AoA, respectively. The ISAC system simultaneously pursues three key objectives: (1) localization of the ISAC Tx UAV, (2) uplink communication, and (3) bi-static sensing of UAV swarm targets using reflected signals. 

In order to reduce hardware cost, the novel spherical DCAA architecture is applied at the Rx, which contains $N$ sUPAs sticking on the surface of a sphere as shown in Fig. \ref{fig:DCAA_env}. Each sUPA has $M\times M$ antenna elements separated by a distance $d=\lambda/2$ along both vertical and horizontal axes, with $\lambda$ being the signal wavelength. Different from the conventional UPA, all the $M\times M$ antenna elements of each sUPA are directly connected without any analog or digital beamforming, which justifies the term sUPA. In addition, with $N_{\mathrm{RF}}\ll N$ RF chains, a selection network is designed to select $N_{\mathrm{RF}}$ ports from the $N$ equivalent sUPA outputs.  The selection network is denoted by $\mathbf{S} =\{ 0,1\}^{N_{\mathrm{RF}}\times N}$, which satisfies $\| [\mathbf{S}]_{i,:} \|=1$ and $[\mathbf{S}]_{:,j}\le 1$, $1\le i \le N_{\mathrm{RF}}$. After that, the $N_{\mathrm{RF}}$ selected outputs will be connected to the RF-chains for subsequent baseband digital processing for communication data decoding and sensing parameter estimation.

Let $x_s$ denote the ISAC signal emitted from the transmitter with transmitting power $P_t=\mathbb{E}\big[|x_s|^2\big]$. The corresponding received signal of $N$ sUPAs at the Rx, denoted by $\mathbf{\tilde{y}}\in\mathbb{C}^{N\times 1}$, is then given by
\begin{equation}\label{eq:rx antenna 1}
\begin{aligned}
    \mathbf{\tilde{y}}&= \mathbf{h}x_s+\mathbf{z},\\
\end{aligned}
\end{equation}
where $\mathbf{z}\in\mathbb{C}^{N\times 1}$ is the AWGN vector, $\mathbf{h}\in\mathbb{C}^{N\times 1}$ is the equivalent channel between the Tx and the $N$ sUPAs, which can be written as
\begin{equation}
    \mathbf{h}=\sum_{l=0}^{L}\mathbf{h}_l,
\end{equation}
where $\mathbf{h}_l$ denotes the channel of the $l$th path.
After passing through the selection network $\mathbf{S}$, the resulting signal $\mathbf{y}\in \mathbb{C}^{N_{\mathrm{RF}} \times 1}$ can be expressed as
\begin{equation}\label{eq:rx antenna 2}
\begin{aligned}
    \mathbf{y}=\mathbf{S} \mathbf{\tilde{y}}=\mathbf{S} \mathbf{h}x_s+\mathbf{S}\mathbf{z}.\\
\end{aligned}
\end{equation}

In the following, we will introduce the specific design of spherical DCAA, selection network, and derive the equivalent channel vector $\mathbf{h}$.

\section{Spherical DCAA Architecture}\label{sec:DCAA Architecture}
\begin{figure}[!h] 
\centering
 \subfloat[\label{fig:planar array}]{
        \centering \includegraphics[width=0.45\columnwidth]{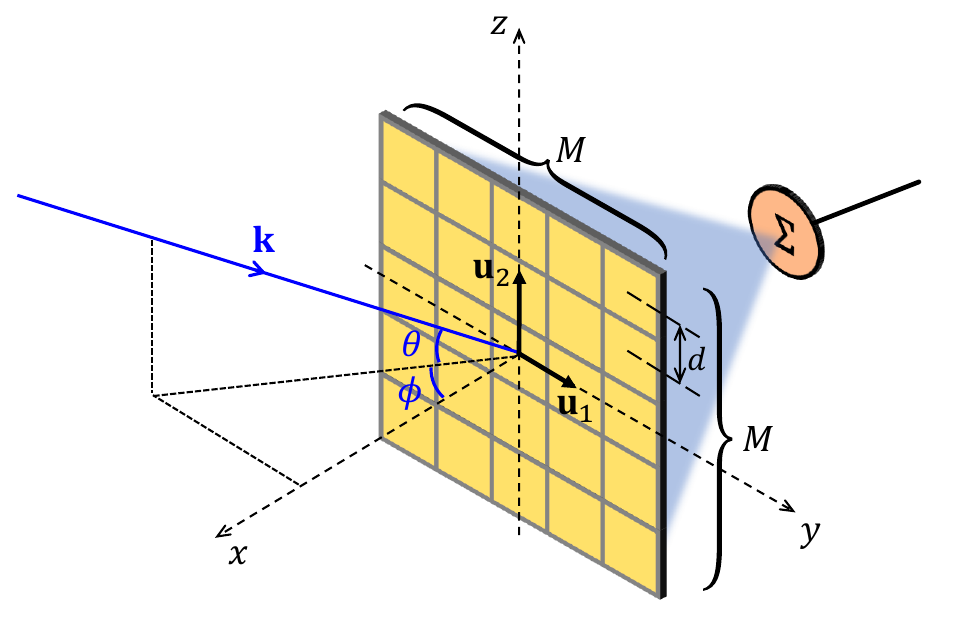}
        }
 \subfloat[\label{fig:planar array after rotation}]{        
        \centering \includegraphics[width=0.45\columnwidth]{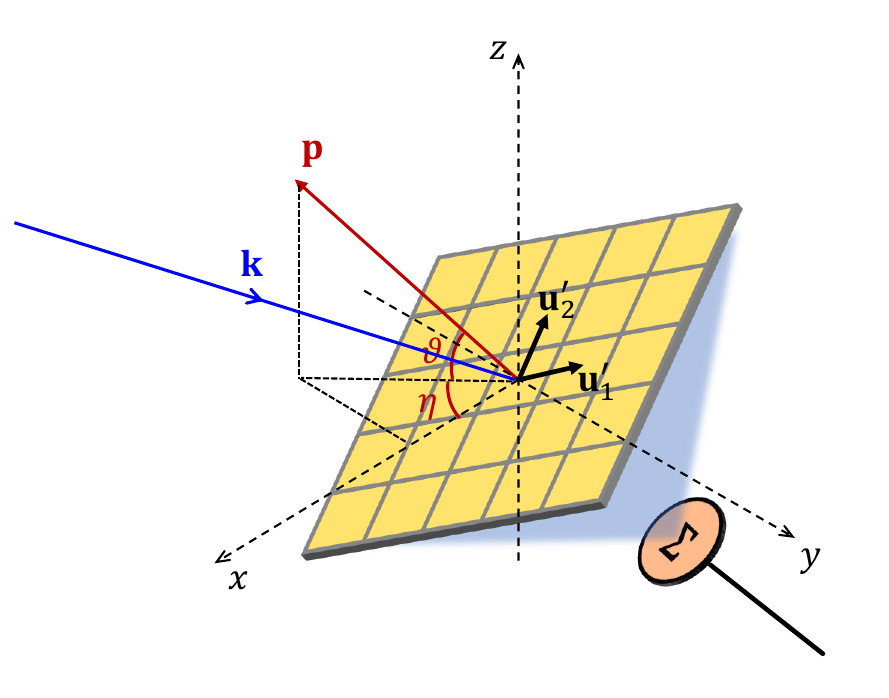}
        }
        \caption{(a) Illustration of sUPA in the proposed spherical DCAA, where all antenna elements of sUPA are separated by half wavelength and directly connected; (b) Rotate the sUPA with respect to $y$-axis and $z$-axis without self-spinning.}
\end{figure}

\subsection{Array Response of Spherical DCAA}
As shown in Fig. \ref{fig:planar array}, consider a $M\times M$ UPA located in the $yz$-plane. Under the assumption of far-field uniform planar wave (UPW), denote the wave vector, azimuth and elevation AoA of the incoming signal as $\mathbf{k}\in\mathbb{C}^{3\times1}$ and $(\phi,\theta)$ respectively, where $\mathbf{k}=[-\cos\theta\cos\phi,-\cos\theta\sin\phi,-\sin\theta]^\mathrm{T}$, $-\pi/2\le\phi,\theta\le\pi/2$. In the figure $\phi$ is measured anti-clockwise from the positive $x$-axis, and $\theta$ is defined relative to the $xy$ plane, taking positive values for directions above the plane and negative values for directions below it.
Let $\mathbf{u}_1=[0,1,0]^\mathrm{T},\mathbf{u}_2=[0,0,1]^\mathrm{T}$ be the normalized vector aligned with the first and second dimension of the UPA. 

To achieve maximum response for signals from multiple directions, the sUPA is rotated, as shown in Fig. \ref{fig:planar array after rotation}. Suppose the rotation of the sUPA does not include any self-spinning, the rotation matrix $\mathbf{R}(\eta,\vartheta )\in\mathbb{C}^{3\times3}$ can be expressed as 
\begin{equation}
    \mathbf{R}(\eta,\vartheta )=\underbrace{\begin{bmatrix}
\cos\eta  & -\sin\eta & 0\\
\sin\eta  & \cos\eta & 0\\
0  & 0 & 1
\end{bmatrix}}_{\mathbf{R}_z(\eta)} \underbrace{\begin{bmatrix}
\cos\vartheta  & 0 & -\sin\vartheta\\
0  & 1 & 0\\
\sin\vartheta  & 0 & \cos\vartheta
\end{bmatrix}}_{\mathbf{R}_y(\vartheta)}, 
\end{equation}
where $\mathbf{R}_z(\eta)\in \mathbb{R}^{3\times3}$ is the yaw matrix rotating around the $z$-axis, with $\eta \in [-\eta_{\max},\eta_{\max}]$ being the yaw angle w.r.t. the positive $y$-axis (anticlockwise on the $xy$-plane); and $\mathbf{R}_y(\vartheta)\in \mathbb{R}^{3\times3}$ is the pitch matrix rotating around the $y$-axis, with $\vartheta\in [-\vartheta_{\max},\vartheta_{\max}]$ being the pitch angle w.r.t. the positive $z$-axis (anticlockwise on the $xz$-plane), with $\vartheta_{\max} < \pi/2$. Let the normal vector of the plane containing the antenna array be denoted by $\mathbf{p}=[-\cos\vartheta\cos\eta,-\cos\vartheta\sin\eta,-\sin\vartheta]^\mathrm{T}\in\mathbb{C}^{3\times 1}$. Thus, we term the rotated sUPA as the sUPA with orientation $(\eta,\vartheta)$, and the normalized side vector of the rotated sUPA can be expressed as 
\begin{equation}
\begin{aligned}
     &\mathbf{u}_1'=\mathbf{R}(\eta,\vartheta )\mathbf{u}_1=\begin{bmatrix}
 -\sin\eta,\cos\eta,0
\end{bmatrix}^T,\\
  & \mathbf{u}_2'=\mathbf{R}(\eta,\vartheta )\mathbf{u}_2=\begin{bmatrix}
 -\cos\eta\sin\vartheta , -\sin\eta\sin\vartheta ,\cos\vartheta
\end{bmatrix}^\mathrm{T}.
\end{aligned}
\end{equation}

The array response vector of the rotated UPA with incoming wave vector $\mathbf{k}$ can be obtained
\begin{equation}\label{eq:sUPA steer vector}
    \mathbf{a}(\mathbf{u}_1',\mathbf{u}_2',\mathbf{k})=\sqrt{G(\mathbf{u}_1',\mathbf{u}_2',\mathbf{k})} \mathbf{a}_1(\mathbf{u}_1',\mathbf{k})\otimes \mathbf{a}_2(\mathbf{u}_2',\mathbf{k}),
\end{equation}
where $G(\mathbf{u}_1',\mathbf{u}_2',\mathbf{k})$ accounts the radiation pattern of each antenna element in the sUPA, which can be equivalently written as $G(\eta-\phi,\vartheta-\theta) $; and $\mathbf{a}_1(\mathbf{u}_1',\mathbf{k})$ and $\mathbf{a}_2(\mathbf{u}_2',\mathbf{k})$ are the steering vector along the side vectors, respectively, given by
\begin{equation}
\begin{aligned}
\mathbf{a}_1(\mathbf{u}_1',\mathbf{k})=\Big[1,e^{-j\pi\mathbf{u}_1^\mathrm{'T}\mathbf{k}},...,e^{-j(M-1)\pi\mathbf{u}_1^\mathrm{'T}\mathbf{k}}\Big]^\mathrm{T},\\
\mathbf{a}_2(\mathbf{u}_2',\mathbf{k})=\Big[1,e^{-j\pi\mathbf{u}_2^\mathrm{'T}\mathbf{k}},...,e^{-j(M-1)\pi\mathbf{u}_2^\mathrm{'T}\mathbf{k}}\Big]^\mathrm{T}.
\end{aligned}
\end{equation}

For sUPA, all the $M\times M$ antenna elements are directly connected. Then the resulting response of the sUPA with orientation $(\mathbf{u}_1', \mathbf{u}_2')$ for incoming signal $\mathbf{k}$, denoted by $r(\mathbf{u}_1',\mathbf{u}_2',\mathbf{k})$, is given by 
\begin{equation}\label{eq:sUPA response with rotate}
\begin{aligned}
&r(\mathbf{u}_1',\mathbf{u}_2',\mathbf{k})=\mathbf{1}^\mathrm{T}_{M^2\times1}\mathbf{a}(\mathbf{u}_1',\mathbf{u}_2',\mathbf{k})\\
&=\sqrt{G(\mathbf{u}_1',\mathbf{u}_2',\mathbf{k})} \underbrace{M^2H_M\big(-\mathbf{u}_1^\mathrm{'T}\mathbf{k}\big)H_M\big(-\mathbf{u}_2^\mathrm{'T}\mathbf{k}\big)}_{f(\mathbf{u}_1',\mathbf{u}_2',\mathbf{k})},
\end{aligned}
\end{equation}
where $H_M(x)=\frac{1}{M}\sum_{m=0}^{M-1}e^{j\pi mx}$ is the Dirichlet kernel
function, given by
\begin{equation}\label{eq:Dirichlet kernel}
H_M(x)=e^{j\frac{\pi}{2}(M-1)x}\frac{\sin(\frac{\pi}{2}Mx)}{M\sin(\frac{\pi}{2}x)}.
\end{equation}



\begin{figure}[!h]
  \centering
  \subfloat[\label{fig:rotate nulls}]{
    \includegraphics[width=0.45\linewidth]{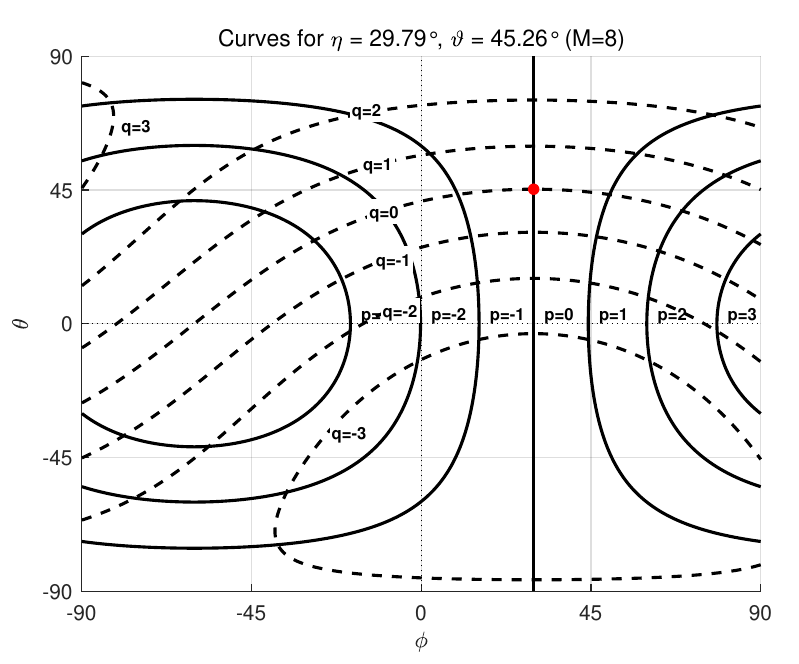}
  }
  \subfloat[\label{fig:rotate beampattern}]
  {
    \includegraphics[width=0.45\linewidth]{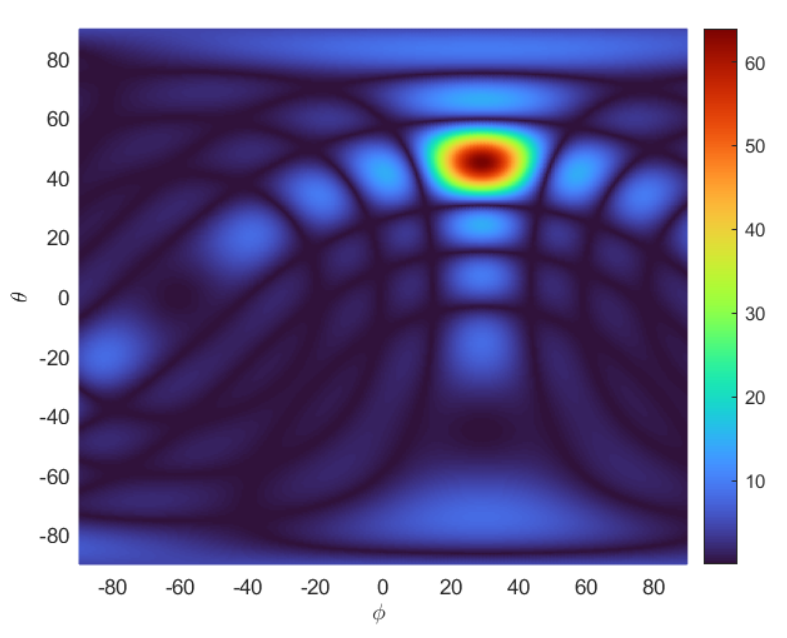}
  }\\
  \caption{Null curves and beam pattern of the sUPA when $M=8$ and $(\eta,\vartheta) = (29^\circ,45^\circ)$. The null curves are plotted through solid lines and dashed lines representing different equations in \eqref{eq:null 1} }
  \label{fig:DCAA nulls}
\end{figure}


Equation \eqref{eq:sUPA response with rotate} can also be written as 
\begin{equation}
\begin{aligned}
&r(\phi,\theta;\eta,\vartheta)\\
&=\sqrt{G( \eta-\phi,\vartheta-\theta )} M^2H_M\big(\cos\theta\sin(\phi-\eta)\big)\\&\quad\times H_M\big(\sin\theta\cos\vartheta-\cos\theta\sin\vartheta\cos(\phi-\eta)\big)\\
&\triangleq\sqrt{G( \eta-\phi,\vartheta-\theta )}f(\phi,\theta;\eta,\vartheta),
\end{aligned}
\end{equation}
where $G(\mathbf{u}_1',\mathbf{u}_2',\mathbf{k})\triangleq G( \eta-\phi,\vartheta-\theta )$. Based on this, we have the following lemmas.

\begin{lemma}\label{lemma:rotate_peak}
For a sUPA with orientation $\mathbf{u}_1',\mathbf{u}_2'$ and incoming wave vector $\mathbf{k}$, we have $\max\limits_{\mathbf{k}}  |f(\mathbf{u}_1',\mathbf{u}_2',\mathbf{k})|=M^2$, and the maximum is achieved if and only if $\mathbf{R}^\mathrm{T}(\eta,\vartheta )\mathbf{k}=[1,0,0]^\mathrm{T}$ or equivalently $(\phi, \theta)=(\eta,\vartheta )$
\end{lemma}

\begin{IEEEproof}
According to equation \eqref{eq:sUPA response with rotate} and \eqref{eq:Dirichlet kernel}, $|f(\mathbf{u}_1',\mathbf{u}_2',\mathbf{k})|$ achieves its maximum only if 
$\mathbf{u}_1^\mathrm{'T}\mathbf{k}=\mathbf{u}_2^\mathrm{'T}\mathbf{k}=0$ and $\mathbf{k}=\mathbf{u}_1'\times\mathbf{u}_2'=\mathbf{R}(\eta,\vartheta )[1,0,0]^\mathrm{T}$.
\end{IEEEproof}


\begin{lemma}\label{lemma:sUPA null with rotate}
For a given sUPA with $\mathbf{u}_1'$ and $\mathbf{u}_2'$ and incoming wave vector $\mathbf{k}$, we have $\min\limits_{\mathbf{k}}  |f(\mathbf{u}_1',\mathbf{u}_2',\mathbf{k})|=0$, and the minimum is achieved if and only if $\mathbf{u}_1^\mathrm{T}\mathbf{R}^\mathrm{T}(\eta,\vartheta )\mathbf{k}=\frac{2p}{M}$ or $\mathbf{u}_2^\mathrm{T}\mathbf{R}^\mathrm{T}(\eta,\vartheta )\mathbf{k}=\frac{2q}{M}$, where $p,q\in\mathbb{Z} \setminus  \left \{ 0 \right \}$ and $|p|\le\frac{M}{2},|q|\le M$.
\end{lemma}

\begin{IEEEproof}
According to equation \eqref{eq:sUPA response with rotate} and \eqref{eq:Dirichlet kernel}, $|f(\mathbf{u}_1',\mathbf{u}_2',\mathbf{k})|$ achieves its minimum if and only if 
$\mathbf{u}_1^\mathrm{'T}\mathbf{k}=\mathbf{u}_1^\mathrm{T}\mathbf{R}^\mathrm{T}(\eta,\vartheta )\mathbf{k}=\frac{2p}{M}$ or $\mathbf{u}_2^\mathrm{'T}\mathbf{k}=\mathbf{u}_2^\mathrm{T}\mathbf{R}^\mathrm{T}(\eta,\vartheta )\mathbf{k}=\frac{2q}{M}$, where $p,q\in\mathbb{Z} \setminus  \left \{ 0 \right \}$ and $|p|\le\frac{M}{2},|q|\le M$.
\end{IEEEproof}

Let \(\theta_{\mathrm{null}}\) and \(\phi_{\mathrm{null}}\) denote the elevation and azimuth angles when the null condition \(|f(\mathbf{u}_1',\mathbf{u}_2',\mathbf{k})| = 0\) holds. According to Lemma \ref{lemma:sUPA null with rotate}, it yields
\begin{equation}\label{eq:null 1}
\begin{aligned}
    &\cos\theta_{\mathrm{null}}\sin(\phi_{\mathrm{null}}-\eta)=\frac{2p}{M},\\
\text{or } 
    &\sin\theta_{\mathrm{null}}\cos\vartheta-\cos\theta_{\mathrm{null}}\sin\vartheta\cos(\phi_{\mathrm{null}}-\eta)=\frac{2q}{M},
    \end{aligned}
\end{equation}
where $p,q\in\mathbb{Z} \setminus  \left \{ 0 \right \}$ and $|p|\le\frac{M}{2},|q|\le M$.

The corresponding null curves and beam pattern are shown in Fig. \ref{fig:rotate nulls} and \ref{fig:rotate beampattern}. While a single sUPA produces a beam pattern with a continuous network of nulls that creates a web of blind areas, Lemma \ref{lemma:rotate_peak} indicates that this property can be harnessed as an advantage. By confining the maximum response to a specific direction, it simultaneously enhances communication and sensing through energy focusing and interference suppression. The null-to-null beamwidth of the spherical DCAA main lobe can be obtained, which is given in Lemma \ref{lemma:sUPA beamwidth}.

\begin{lemma}\label{lemma:sUPA beamwidth}
For any given sUPA with orientation $(\eta,\vartheta)$, the null-to-null azimuth and elevation beamwidth of the main lobe, denoted by $\textbf{BW}(\vartheta,\eta)=(\Delta\phi,\Delta\theta)$, is expressed as
\begin{equation}
    \textbf{BW}(\vartheta,\eta)=\bigg(2\arcsin\Big(\frac{2}{M\cos\vartheta}\Big),2\arcsin\Big(\frac{2}{M}\Big)\bigg).
\end{equation}
\end{lemma}

\begin{IEEEproof}
    Please refer to Appendix \ref{app:sUPA beamwidth}
\end{IEEEproof}

Lemma \ref{lemma:sUPA beamwidth} indicates that the sUPA's main lobe has a uniform beamwidth $2\arcsin(2/M)$ in the elevation domain but a non-uniform one in the azimuth domain, and the latter increases with the elevation angle $\vartheta$. This characteristic ensures a uniform elevation angular resolution for the spherical DCAA and, furthermore, it implies that the sUPA becomes effectively sparser at larger elevation angles.

\subsection{sUPA Orientation Design and Spherical DCAA Implementation}
Under the premise of covering the entire space and minimize interference between adjacent sUPAs in the elevation and azimuth directions, the sUPAs are strategically arranged with uniform spacing of $\arcsin\frac{2}{M\cos\vartheta}$ and $\arcsin\frac{2}{M}$ respectively in the azimuth and elevation directions, and we define $N_\vartheta=\left \lfloor \theta_{\max}/\arcsin\frac{2}{M} \right \rfloor$ and $N_{\eta}(\vartheta)=\left \lfloor \phi_{\max}/\arcsin\frac{2}{M\cos\vartheta} \right \rfloor$ as the number of layers of sUPAs and the number of sUPAs for $\vartheta$, respectively, where $\theta_{\max}$ and $\phi_{\max}$ are the maximum elevation and azimuth angles of the signals. Therefore, with the original sUPA with $(\eta,\vartheta)=(0,0)$, we adopt the following two steps to determine the directions of all the $N$ sUPAs.
\begin{itemize}
    \item \textit{Elevation:} We first find all the sUPA directions $(\eta, \vartheta)$ in the elevation dimension with $\eta=0$, which yields
    \begin{equation}\label{eq:thetaq}
        \vartheta_q=q\arcsin(2/M),q\in\mathcal{N}_{\vartheta},
    \end{equation}
    where $\mathcal{N}_{\vartheta}=\{-N_\vartheta,...,0,1,...,N_\vartheta\}$.
    \item \textit{Azimuth:} For elevation angle $\vartheta_q$, we identify all corresponding azimuth angles $\eta_{p,q}$ in the azimuth domain.
    \begin{equation}\label{eq:etapq}
    \begin{aligned}
         (\eta_{p,q},\vartheta_{q})=\bigg(p\arcsin\Big(\frac{2}{M\cos\vartheta_q}\Big),\vartheta_q\bigg),p\in\mathcal{N}_{\eta}(\vartheta_q),
    \end{aligned}
    \end{equation}
    where $\mathcal{N}_{\eta}(\vartheta_q)=\{-N_{\eta}(\vartheta_q),...,0,1,...,N_{\eta}(\vartheta_q)\}$.
\end{itemize}

The total number $N$ of sUPAs is thus
\begin{equation}\label{eq:number of N}
    N=\sum_{q=-N_\vartheta}^{N_\vartheta}N_{\eta}(\vartheta_q).
\end{equation}





\begin{figure}[t] 
        \centering \includegraphics[width=0.6\columnwidth]{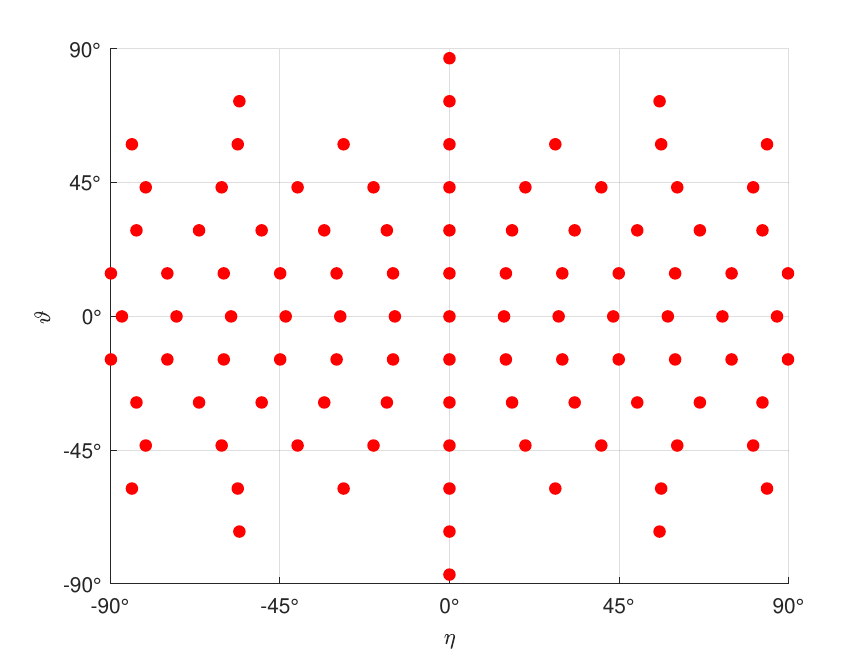}
        \caption{\label{fig:setup1}The orientation of sUPAs $(\eta_{p,q},\vartheta_{q})$ according to equation \eqref{eq:etapq}, where $M=8$.}
\end{figure}


\begin{figure}[t]
  \centering
  \subfloat[\label{fig:beam1}]{
    \includegraphics[width=0.42\linewidth]{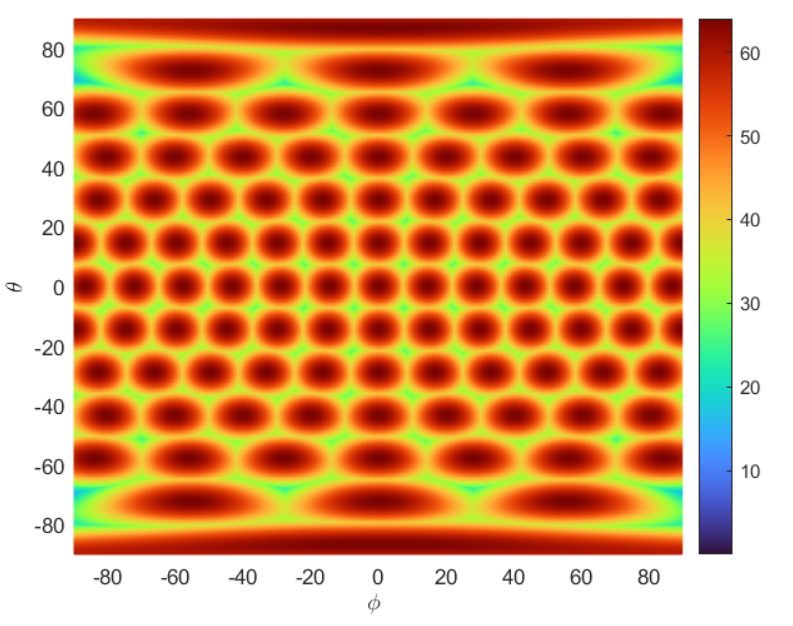}
  }
  \subfloat[\label{fig:KPC_beam}]
  {
    \includegraphics[width=0.42\linewidth]{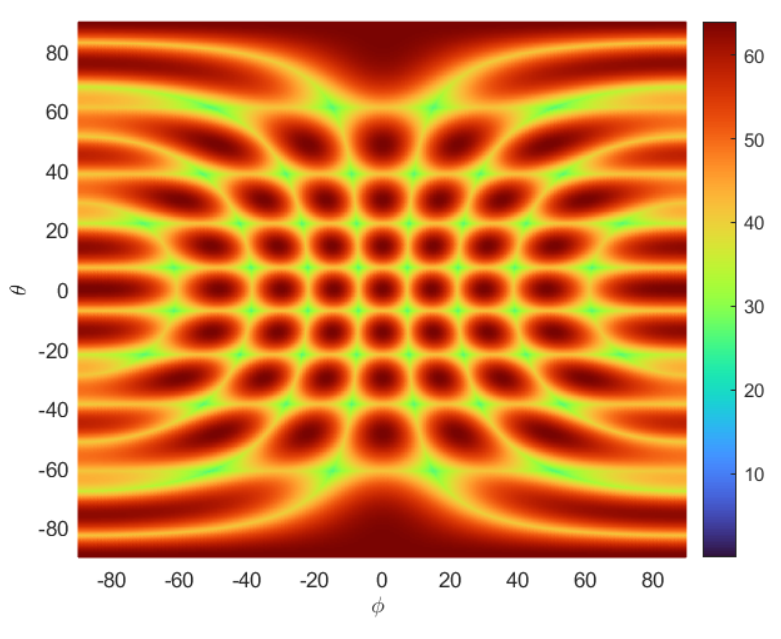}
  }\\
  \subfloat[\label{fig:maxbeam1}]
  {
    \includegraphics[width=0.42\linewidth]{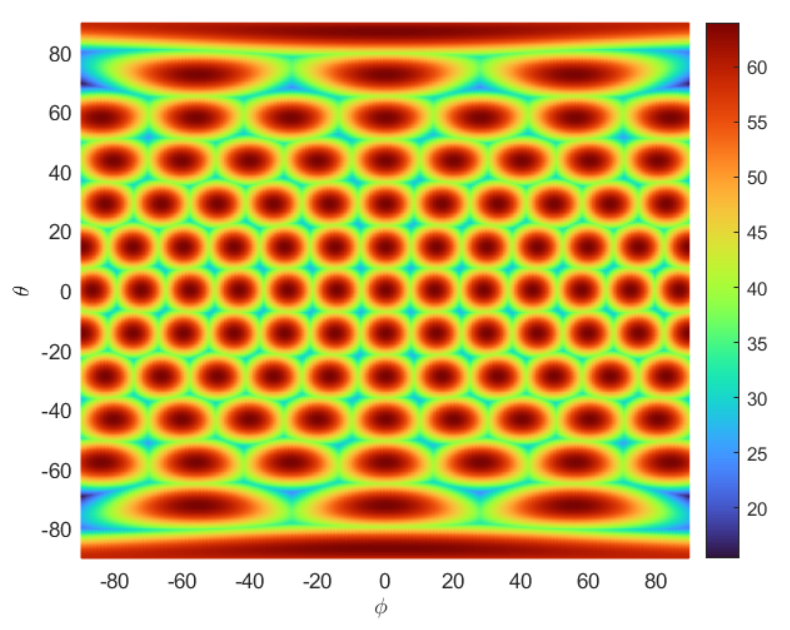}
  }
  \subfloat[\label{fig:KPC_maxbeam}]
  {
    \includegraphics[width=0.42\linewidth]{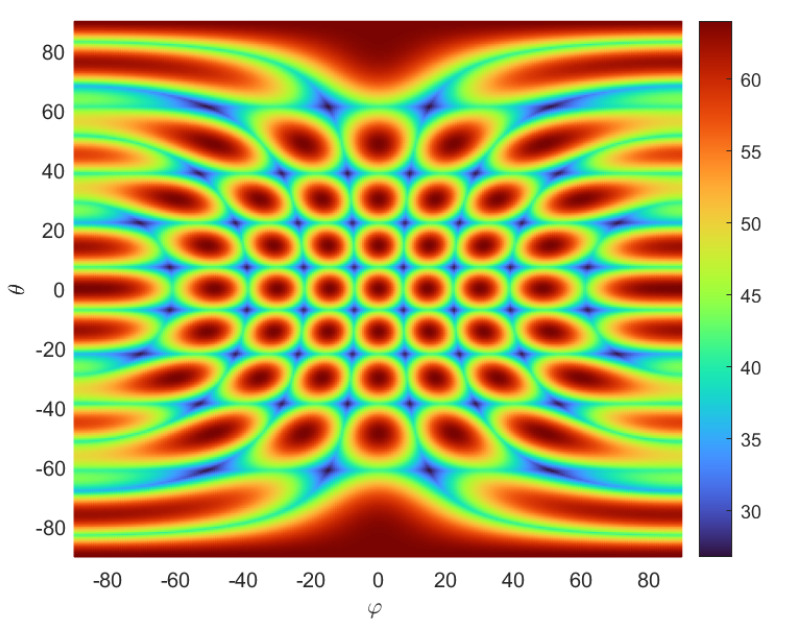}
  }
  
  \caption{The beam pattern of: (a) spherical DCAA; (b) conventional UPA with KPC-based HBF, and the envelope of beam pattern of: (c) spherical DCAA; (d) conventional UPA with KPC-based HBF.}
  \label{fig:beam pattern comparision}
\end{figure}




For ease of presentation, we transform the two-dimensional indexed orientation $(\eta_{p,q},\vartheta_{q}),p\in\mathcal{N}_{\eta}(\vartheta_q),q\in\mathcal{N}_{\vartheta}$ into one-dimensional indexed orientation $(\eta_{n},\vartheta_{n})$, where $n\in\mathcal{N}=\{1,2,...,N\}$, which satisfies
\begin{equation}
    n = p + N_{\eta}(\vartheta_q)+1+\sum_{k=-N_\vartheta}^{q-1} \Big(2N_{\eta}(\vartheta_k) + 1 \Big),
\end{equation}

The elevation and azimuth orientations of sUPAs with isotropic antenna elements and corresponding beam pattern of the setup are illustrated in Fig. \ref{fig:setup1} and \ref{fig:beam1}. The spherical DCAA exhibits a beam pattern with densely packed yet well-separated peaks. In each direction, the maximum magnitude (envelope) received among all the sUPAs, i.e., $r^*(\phi,\theta)=\max\limits_{n}r(\phi,\theta;\eta_{n},\vartheta_{n})$ is shown in Fig. \ref{fig:maxbeam1}. With this design, the entire spatial domain is covered without any holes, as every direction is served by at least one sUPA. This is quantified by the worst-case coverage metric $\min\limits_{\phi,\theta} r^*(\phi,\theta) > 15$ for $M=8$.

\begin{figure}[t]
  \centering
  \subfloat[\label{fig:implement1}]{
    \includegraphics[width=0.38\linewidth]{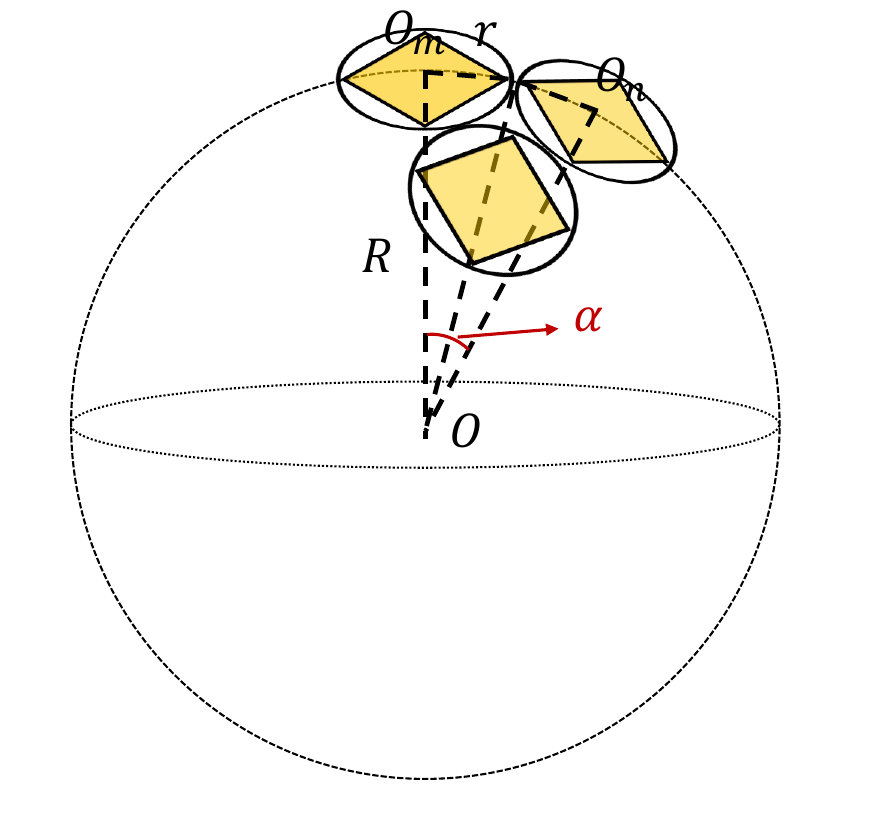}
  }
  \subfloat[\label{fig:implement2}]
  {
    \includegraphics[width=0.48\linewidth]{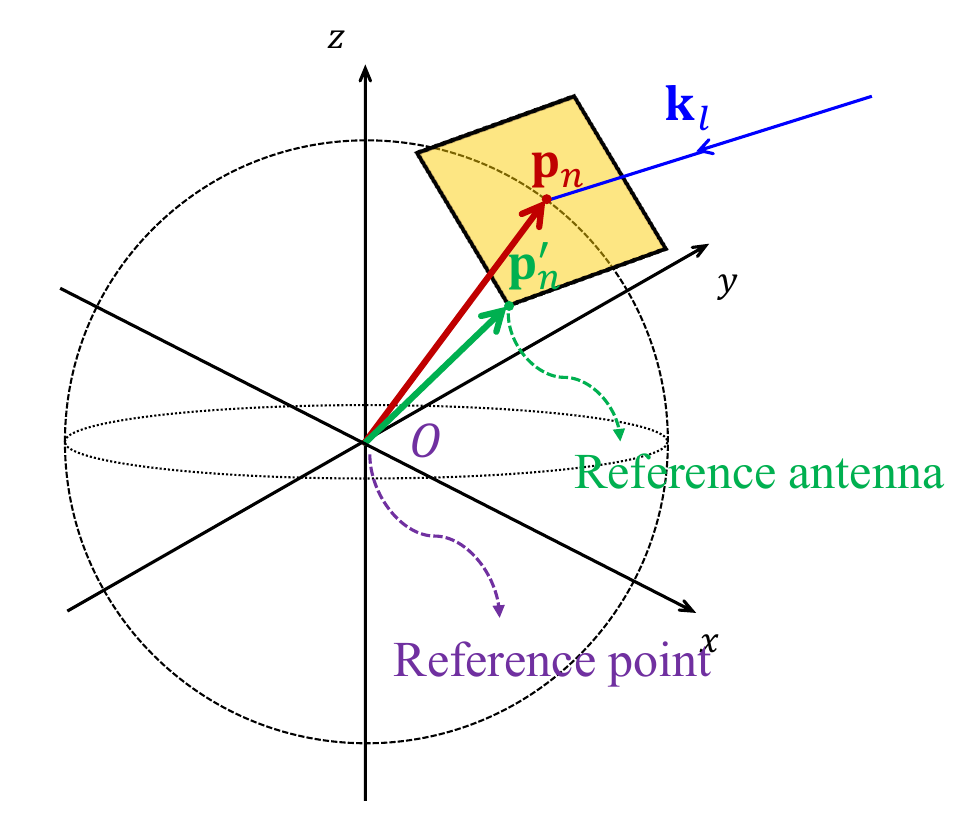}
  }
  
  \caption{Implementation of spherical DCAA (a) the critical collision state where all the sUPAs are closest to each other (b) the geometry relationship among sUPA center location $\mathbf{p}_n$, reference antenna location $\mathbf{p}_n'$, global reference point $O$ and the $l$th incoming signal vector $\mathbf{k}_l$.}
  \label{fig:implement}
\end{figure}


As shown in Fig. \ref{fig:implement}, in order to avoid signal blockage among sUPAs, we arrange all the sUPAs on a sphere with center $O$, where each sUPA is tangent to the spherical surface with tangent point $O_m,O_n$ etc. In addition, we consider the hitbox of sUPAs as a circumcircle with radius $r=Md/\sqrt{2}$ of the square shape of the array. The angle separation $\alpha$ between any two sUPAs is defined as the angle of the two vectors $\overrightarrow{OO_m}$ and $\overrightarrow{OO_n}$, 
and thus we have the following theorem.
\begin{theorem}\label{the:angle_separation}
In the aforementioned design, the minimum angle separation, denoted by $\alpha^+$, between any two sUPAs is $\arcsin\frac{2}{M}$.
\end{theorem}

\begin{IEEEproof}
    Please refer to Appendix \ref{app:minimum angle}
\end{IEEEproof}

In the critical collision state, the circumcircles of the sUPAs are tangent to each other as shown in Fig. \ref{fig:implement1}, and the minimum angle separation $\alpha^*$ between any two sUPAs can be obtained as
\begin{equation}
    \tan\frac{\alpha^*}{2}=\frac{r}{R}=\frac{Md}{\sqrt{2}R}.
\end{equation}
Therefore, the minimum angle separation $\alpha^+$ in design must be greater than the critical collision state minimum angle separation $\alpha^*$, i.e., $\alpha^+\ge\alpha^*$, and the minimum radius $R$ of the sphere can be obtained
\begin{equation}\label{eq:radius}
    R=\frac{Md}{\sqrt{2}\tan(\frac{1}{2}\arcsin\frac{2}{M})}\overset{(a)}{\approx}\frac{M^2\lambda}{2\sqrt{2}},
\end{equation}
where $(a)$ holds when $M$ is large.
The relationship between $R,M$ and carrier frequency is shown in Fig. \ref{fig:radius}. The minimum radius is proportional to the total number of antenna elements, and it decreases as the carrier frequency grows up. The complete procedure for designing the parameters $\big\{M^2,N,R,\{\eta_{n}, \vartheta_{n}\}_{n\in\mathcal{N}}\big\}$ of the proposed spherical DCAA architecture is illustrated in Fig. \ref{fig:flowchart}.

\begin{figure}[htbp]
        \centering \includegraphics[width=0.65\columnwidth]{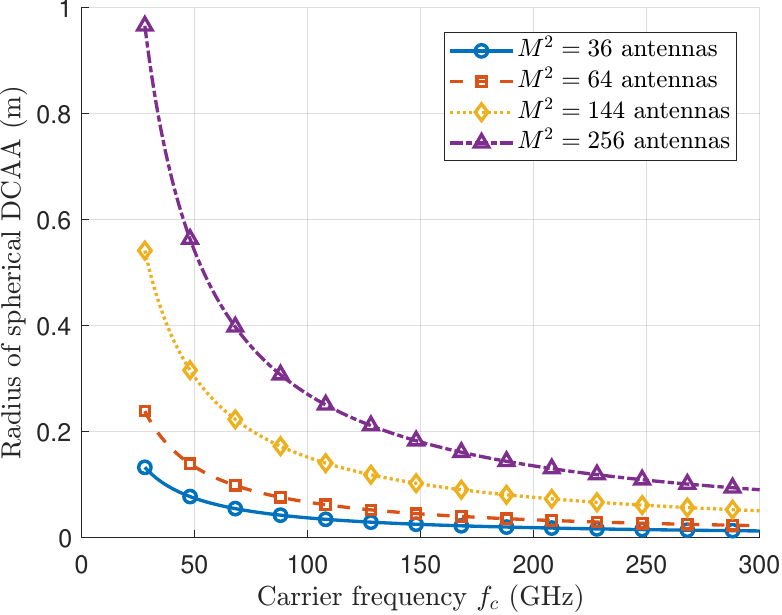}
        \caption{\label{fig:radius}Illustration of the minimum sphere radius needed to implement DCAA w.r.t. carrier frequency and number of antenna elements.}
\end{figure}

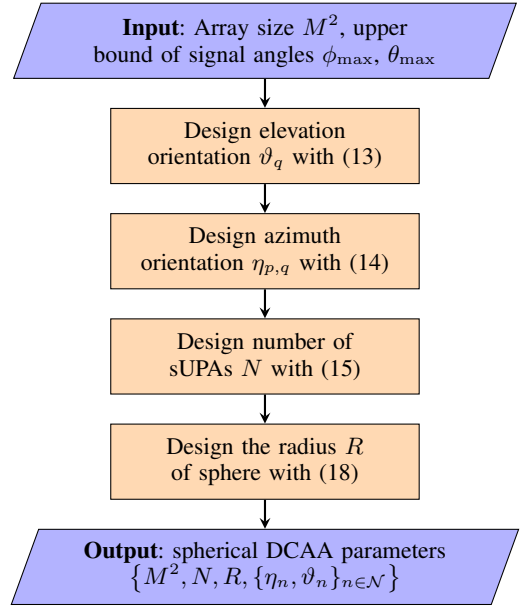
\begin{figure}[htbp]\label{fig:flowchart}
\centering
\tikzset{every node/.style={font=\small}} 
\begin{tikzpicture}[node distance=1.4cm]
\node (in) [io, text width=15em] {\textbf{Input}: Array size $M^2$, upper bound of signal angles $\phi_{\max}$, $\theta_{\max}$};
\node (pro1) [process, text width=11em, below of=in] {Design elevation orientation $\vartheta_q$ with \eqref{eq:thetaq}};
\node (pro2) [process, text width=11em, below of=pro1] {Design azimuth orientation $\eta_{p,q}$ with \eqref{eq:etapq}};
\node (pro3) [process, text width=11em, below of=pro2] {Design number of sUPAs $N$ with \eqref{eq:number of N}};
\node (pro4) [process, text width=11em, below of=pro3] {Design the radius $R$ of sphere with \eqref{eq:radius}};
\node (out) [io, text width=15em, below of=pro4] {\textbf{Output}: spherical DCAA parameters $\big\{M^2,N,R, \{\eta_{n}, \vartheta_{n}\}_{n\in\mathcal{N}}\big\}$};

\draw [arrow] (in) -- (pro1);
\draw [arrow] (pro1) -- (pro2);
\draw [arrow] (pro2) -- (pro3);
\draw [arrow] (pro3) -- (pro4);
\draw [arrow] (pro4) -- (out);
\end{tikzpicture}
\caption{Flowchart of spherical DCAA parameter design}
\label{fig:flowchart}
\end{figure}

As a result, for the UAV swarm ISAC system with one LoS component and $L$ NLoS components, the equivalent channel $\mathbf{h}$ between the Tx and the $N$ sUPAs of the spherical DCAA can be expressed as
\begin{equation}\label{eq:channel}
    \mathbf{h}=\sum_{l=0}^L\alpha_l\mathbf{r}(\phi_l,\theta_l),
\end{equation}
where $\alpha_l$ denotes the path coefficient; $\phi_l,\theta_l$,  denote the path azimuth and elevation AoAs respectively; and $\mathbf{r}(\phi_l,\theta_l)=[r(\phi,\theta;\eta_n,\vartheta_n)]_{n\in\mathcal{N}}\in\mathbb{C}^{N\times 1}$ denotes the array response vector of spherical DCAA.

As shown in Fig. \ref{fig:implement2}, we denote $\mathbf{p}_n=[-R\cos\vartheta_n\cos\eta_n,-R\cos\vartheta_n\sin\eta_n,-R\sin\vartheta_n]^\mathrm{T}$ as the position vector of the center of the $n$th sUPA; and $\mathbf{p}_n'$ is denoted as the position vector of the reference antenna of the $n$th sUPA, which satisfies
\begin{equation}
    \mathbf{p}_n'=\mathbf{p}_n-\frac{1}{2}\mathbf{R}(\eta_n,\vartheta_n)(\mathbf{u}_1+\mathbf{u}_2)Md. 
\end{equation}

To model the blockage effect, $\epsilon(-\mathbf{k}_l^\mathrm{T}\mathbf{p}_n)$ is introduced, where $\mathbf{k}_l=[-\cos\theta_l\cos\phi_l,-\cos\theta_l\sin\phi_l,-\sin\theta_l]^\mathrm{T}\in\mathbb{C}^{3\times 1}$ is the $l$th incoming signal vector and $\epsilon(\cdot)$ denotes the Heaviside step function which forces the response of a sUPA to zero whenever it is not directly illuminated (where $\mathbf{k}_l^\mathrm{T}\mathbf{p}_n>0$) by the incoming plane wave.
Take the center of the sphere as the reference point, the array response can be written as
\begin{equation}\label{eq:r}
\begin{aligned}
    \mathbf{r}&(\phi_l,\theta_l)=\Big[\epsilon(-\mathbf{k}_l^\mathrm{T}\mathbf{p}_n) e^{-j\frac{2\pi}{\lambda}\frac{\mathbf{k}_l^\mathrm{T}\mathbf{p}_n'}{|\mathbf{p}_n'|}} r(\phi_l,\theta_l;\eta_n,\vartheta_n) \Big]_{n\in\mathcal{N}}\\
    &=\Big[ \epsilon(-\mathbf{k}_l^\mathrm{T}\mathbf{p}_n)e^{-j\frac{2\pi}{\lambda}\frac{\mathbf{k}_l^\mathrm{T}\mathbf{p}_n'}{|\mathbf{p}_n'|}} M^2\sqrt{G( \eta_n-\phi_l,\vartheta_n-\theta_l )}\\&\quad\times H_M\big(\cos\theta_l\sin(\phi_l-\eta_n)\big)\\&\quad\times H_M\big(\sin\theta_l\cos\vartheta_n-\cos\theta_l\sin\vartheta_n\cos(\phi_l-\eta_n)\big)\Big]_{n\in\mathcal{N}}.
\end{aligned}
\end{equation}

Due to the limited RF chains, we need to choose $N_{\mathrm{RF}}$ signals from $N$ spherical DCAA outputs. To fully exploit the energy-focusing ability of spherical DCAA, we propose the selection scheme based on energy maximization, i.e., $\mathbf{S}=\arg\max\limits_{\mathbf{S}} \| \mathbf{y} \|_2^2$, to choose the $N_{\mathrm{RF}}$ spherical DCAA outputs with the maximum sum energy.
To realize the energy based sUPA selection, we only need to sweep all the spherical DCAA ports to obtain full information about the response magnitude of different sUPAs. This would require $\left \lceil \frac{N}{N_{\mathrm{RF}}} \right \rceil$ sweeps to cover all the $N$ sUPAs. Therefore, $\mathbf{S}$ will be determined and fixed for the subsequent analysis.


\section{Spherical DCAA vs Conventional UPA}\label{sec:DCAAvsUPA}
\begin{figure}[htbp]
        \centering \includegraphics[width=1\linewidth]{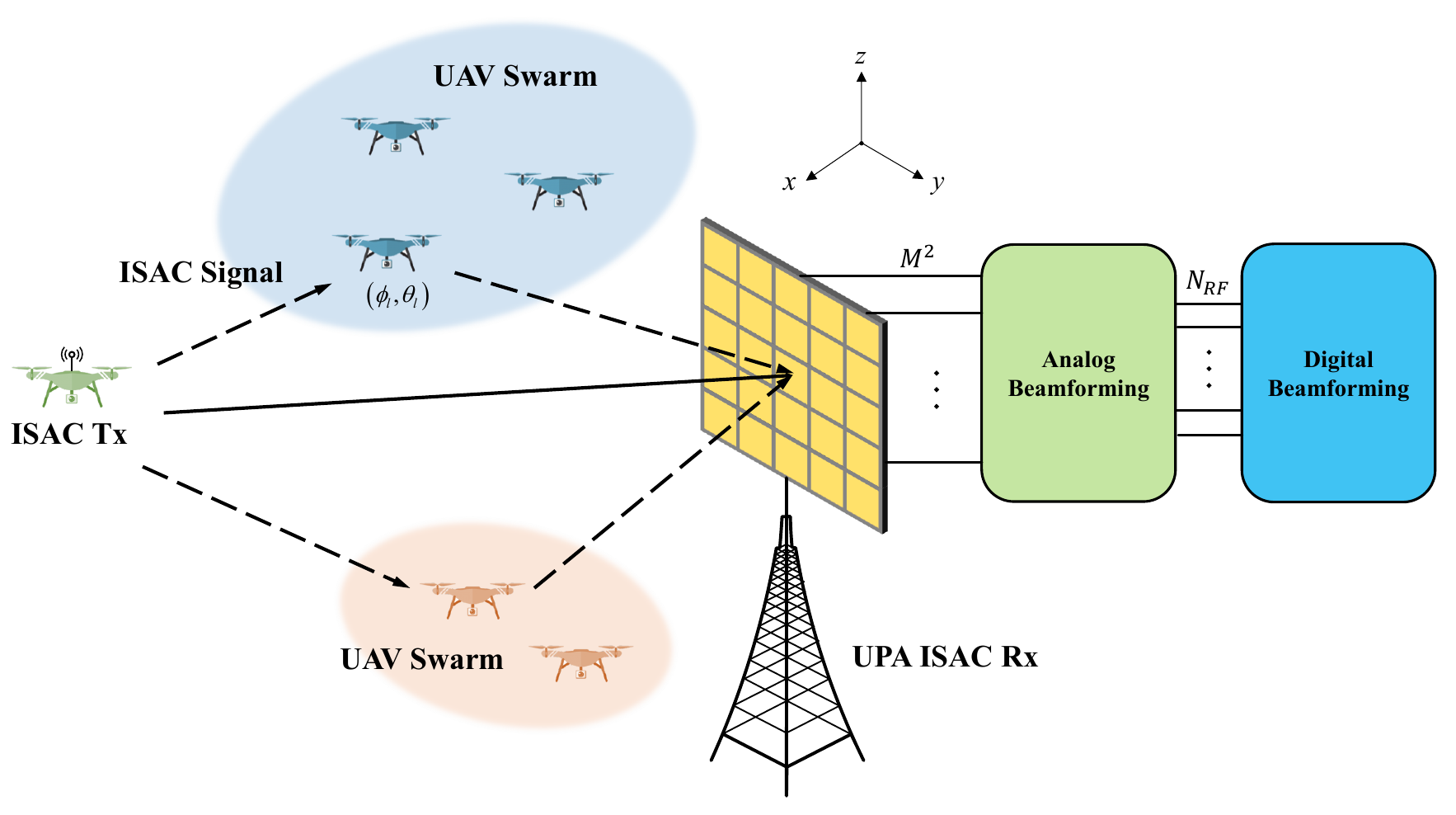}
        \caption{\label{fig:KPC_env}An illustration of conventional UPA-based ISAC for low-altitude UAV swarm.}
\end{figure}

\begin{figure}[t]
        \centering \includegraphics[width=0.95\linewidth]{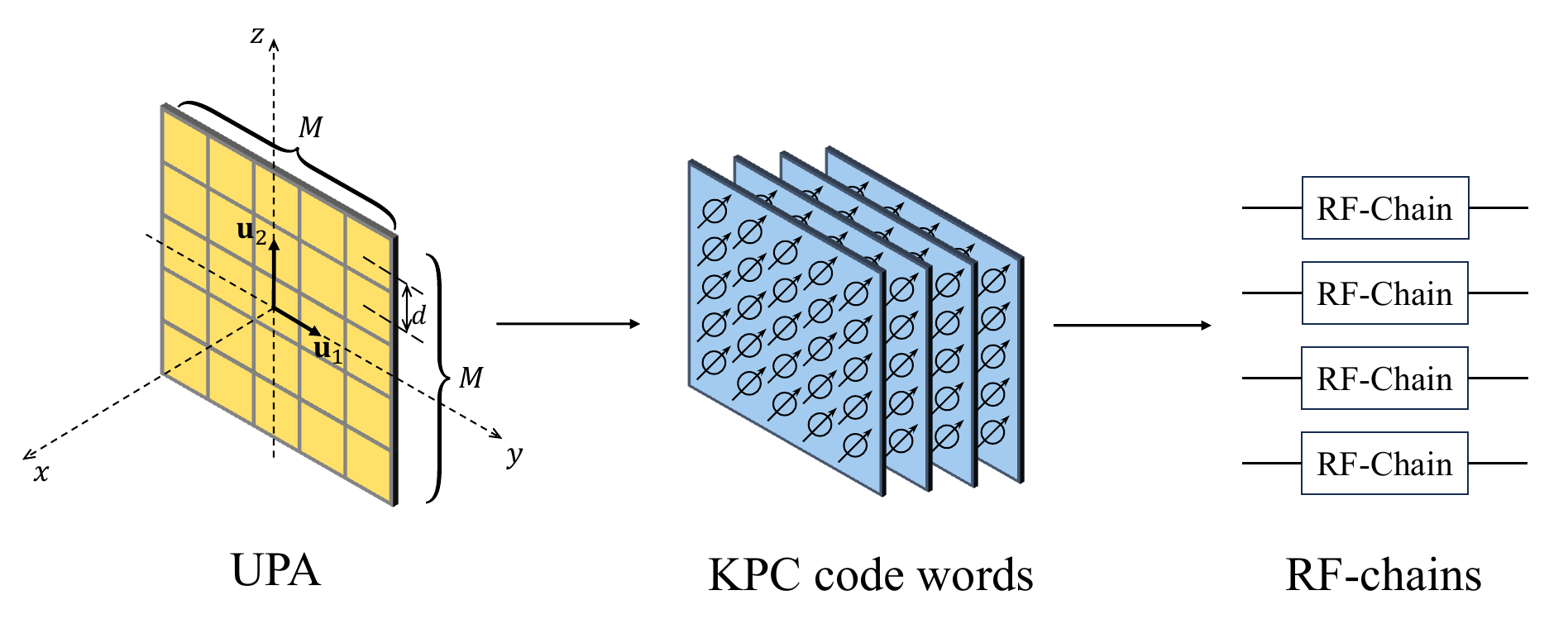}
        \caption{\label{fig:KPC_architecture}Antenna architecture of conventional UPA using the classical KPC-based hybrid beamforming.}
\end{figure}

\subsection{Conventional UPA}
As a comparison, we consider a UPA using KPC-based HBF, as illustrated in Fig. \ref{fig:KPC_env}. A UPA facing $x$-axis with orientation vectors $\mathbf{u}_1$, $\mathbf{u}_2$ and $M\times M$ antenna elements is shown in Fig. \ref{fig:KPC_architecture}. The adjacent antenna elements are separated by half wavelength in both directions. Thus, according to \eqref{eq:sUPA steer vector}, the array response vector of the UPA for the propagation path $\mathbf{k}$ can also be written as $\mathbf{a}(\mathbf{u}_1,\mathbf{u}_2,\mathbf{k})$.

The azimuth and elevation DFT codewords are respectively defined as
\begin{equation}
    \begin{aligned}
        \mathbf{c}_h(\phi_p) &= [1, e^{j\pi\sin\phi_p}, \dots, e^{j\pi(M-1)\sin\phi_p}]^\mathrm{T} \in \mathbb{C}^{M \times 1}, \\
        \mathbf{c}_v(\theta_q) &= [1, e^{j\pi\sin\theta_q}, \dots, e^{j\pi(M-1)\sin\theta_q}]^\mathrm{T} \in \mathbb{C}^{M \times 1},
    \end{aligned}
\end{equation}
where $\sin\phi_p = \frac{2p}{M}$ and $\sin\theta_q = \frac{2q}{M}$. The $(p,q)$-th codeword of the KPC is obtained via the Kronecker product of the azimuth and elevation DFT codewords, i.e.,
\begin{equation}
    \mathbf{c}(\phi_p, \theta_q) = \mathbf{c}_v(\theta_q) \otimes \mathbf{c}_h(\phi_p) \in \mathbb{C}^{M^2 \times 1}.
\end{equation}

The KPC implemented via the HBF architecture is denoted by $\mathbf{C} = [\mathbf{c}(\phi_p, \theta_q)]_{
        -\frac{N_v-1}{2} \le p \le \frac{N_v-1}{2}, \,
        -\frac{N_h-1}{2} \le q \le \frac{N_h-1}{2}
    } \in \mathbb{C}^{M^2 \times N_v N_h},$
where $N_v$ and $N_h$ are the numbers of codewords in the elevation and azimuth codebooks, respectively, given by $N_v = 2\left\lfloor \frac{\sin\theta_{\max}}{2/M} \right\rfloor + 1$ and $N_h = 2\left\lfloor \frac{\sin\phi_{\max}}{2/M} \right\rfloor + 1$. Consequently, the beam pattern for the codeword $\mathbf{c}(\phi_p, \theta_q)$, denoted by $r_{\mathrm{BF}}(\phi, \theta; \phi_p, \theta_q)$, can be expressed as
\begin{equation}\label{eq:KPC response}
\begin{aligned}
    r_{\mathrm{BF}}(\phi, \theta; \phi_p, \theta_q) 
    &= \mathbf{c}^\mathrm{H}(\phi_p, \theta_q) \mathbf{a}(\mathbf{u}_1, \mathbf{u}_2, \mathbf{k}) \\
    &= \sqrt{G(\phi, \theta)} M^2 H_M(\cos\theta \sin\phi - \sin\phi_p) \\
    &\quad \times H_M(\sin\theta - \sin\theta_p) \\
    &\triangleq \sqrt{G(\phi, \theta)} f_{\mathrm{BF}}(\phi, \theta; \phi_p, \theta_q).
\end{aligned}
\end{equation}

For any given KPC codeword $\mathbf{c}(\phi_p,\theta_q)$, we aim to find the  AoA $(\phi,\theta)$ that yields the maximum and minimum response.

\begin{lemma}\label{lemma:KPC peak direction}
For a given KPC codeword $\mathbf{c}(\phi_p,\theta_q)$, we have $\max\limits_{(\phi,\theta)}  f_{\mathrm{BF}}(\phi,\theta;\phi_p,\theta_q)=M^2$, and the maximum is achieved if and only if $(\phi,\theta)=\big(\arcsin\frac{\sin\phi_p}{\cos\theta_q},\theta_q\big)$.
\end{lemma}

\begin{IEEEproof}
According to equation \eqref{eq:KPC response}, $f_{\mathrm{BF}}(\phi,\theta;\phi_p,\theta_q)$ achieves its maximum if and only if $\cos\theta\sin\phi-\sin\phi_p=\sin\theta-\sin\theta_q=0$. Since $|\phi|\le\pi/2,|\theta|\le\pi/2$, we have $(\phi,\theta)=\big(\arcsin\frac{\sin\phi_p}{\cos\theta_q},\theta_q\big)$.
\end{IEEEproof}

\begin{lemma}\label{lemma:KPC null no rotate}
For a given KPC codeword $\mathbf{c}(\phi_p,\theta_q)$, we have $\min\limits_{(\phi,\theta)}  f_{\mathrm{BF}}(\phi,\theta;\phi_p,\theta_q)=0$, and the minimum is achieved if and only if $\theta=\arcsin(\sin\theta_p+\frac{2m}{M})\ \text{with}\ \left|\sin\theta_p + \frac{2m}{M}\right| \le 1$, or $\phi=\arcsin\frac{\sin\phi_p+2n/M}{\cos\theta}\ \text{with}\  \left|\frac{\sin\phi_p + 2n/M}{\cos\theta}\right| \le 1 , \ \cos\theta \neq 0$, where $m,n\in\mathbb{Z} \setminus  \left \{ 0 \right \}$.
\end{lemma}

\begin{IEEEproof}
According to equation \eqref{eq:KPC response}, $f_{\mathrm{BF}}(\phi,\theta;\phi_p,\theta_q)$ achieves its minimum if and only if $\cos\theta\sin\phi-\sin\phi_p=\frac{2n}{M}$ or $\sin\theta-\sin\theta_q=\frac{2m}{M}$, which yields $\theta=\arcsin(\sin\theta_p+\frac{2m}{M})\ \text{with}\ \left|\sin\theta_p + \frac{2m}{M}\right| \le 1$, or $\phi=\arcsin\frac{\sin\phi_p+2n/M}{\cos\theta}\ \text{with}\ \left|\frac{\sin\phi_p + 2n/M}{\cos\theta}\right| \le 1 , \ \cos\theta \neq 0$, where $m,n\in\mathbb{Z} \setminus  \left \{ 0 \right \}$.
\end{IEEEproof}

Lemma \ref{lemma:KPC peak direction} indicates that with codeword $\mathbf{c}(\phi_p,\theta_q)$, the peak of the beam pattern of UPA deviates from the direction of the code word. Specifically, if $\sin\phi_p/\cos\theta_q>1$, the peak becomes meaningless because such an angle $\arcsin\frac{\sin\phi_p}{\cos\theta_q}$ does not exist. 
The beam pattern of KPC are shown in Fig. \ref{fig:KPC_beam}. In each direction, the maximum magnitude received among all the codewords are shown in Fig. \ref{fig:KPC_maxbeam}. Compared with conventional UPA using KPC-based HBF, the proposed spherical DCAA has denser and more regular-shaped peaks in the beam pattern, which can enhance its communication and sensing ability.




\subsection{Angular Resolution Comparison}

\begin{figure*}[t]
  \centering
  \subfloat[\label{fig:RAA_vertical_resolution_top}]{
    \includegraphics[width=0.23\linewidth]{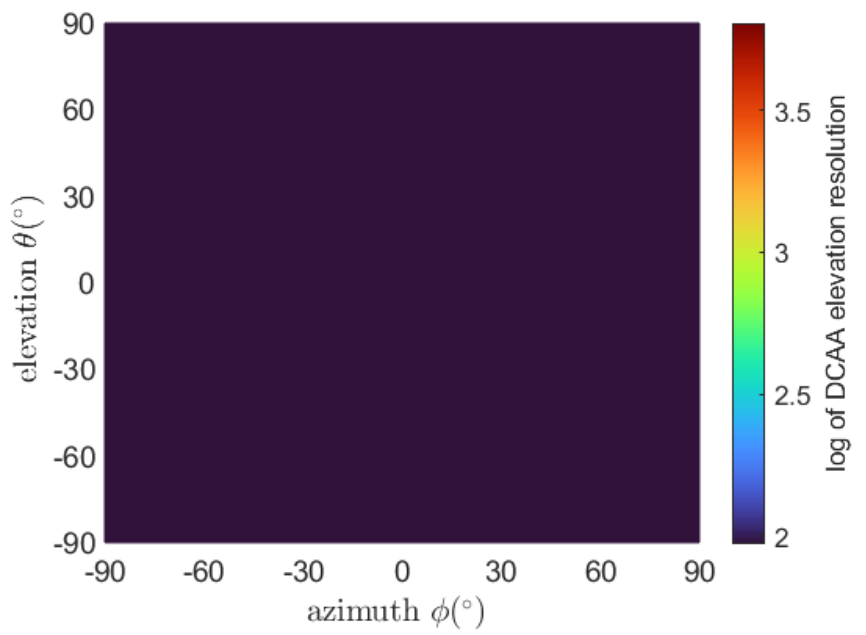}
  }
  \subfloat[\label{fig:KPC_vertical_resolution_top}]{
    \includegraphics[width=0.23\linewidth]{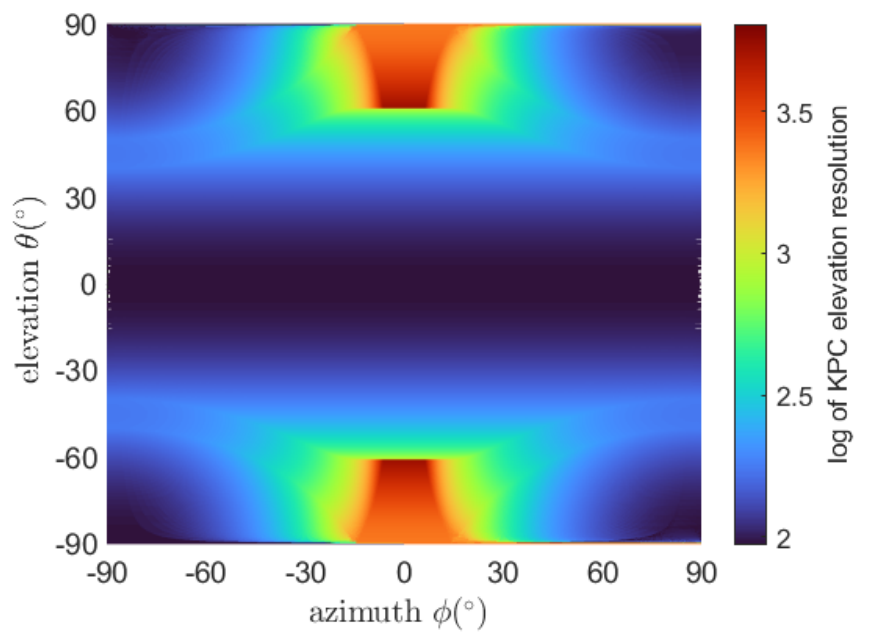}
  }
  \subfloat[\label{fig:RAA_horizontal_resolution_top}]{
    \includegraphics[width=0.23\linewidth]{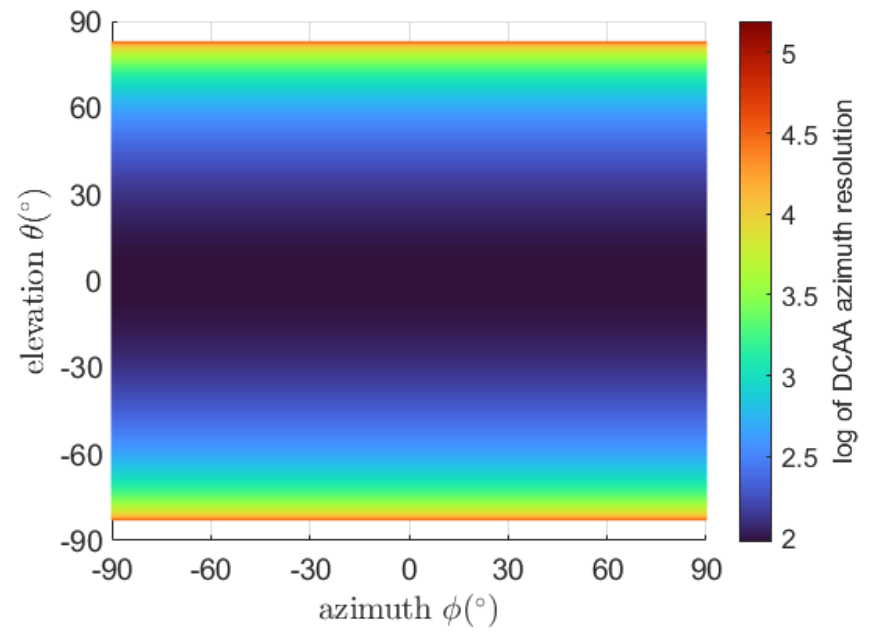}
  }
  \subfloat[\label{fig:KPC_horizontal_resolution_top}]{
    \includegraphics[width=0.23\linewidth]{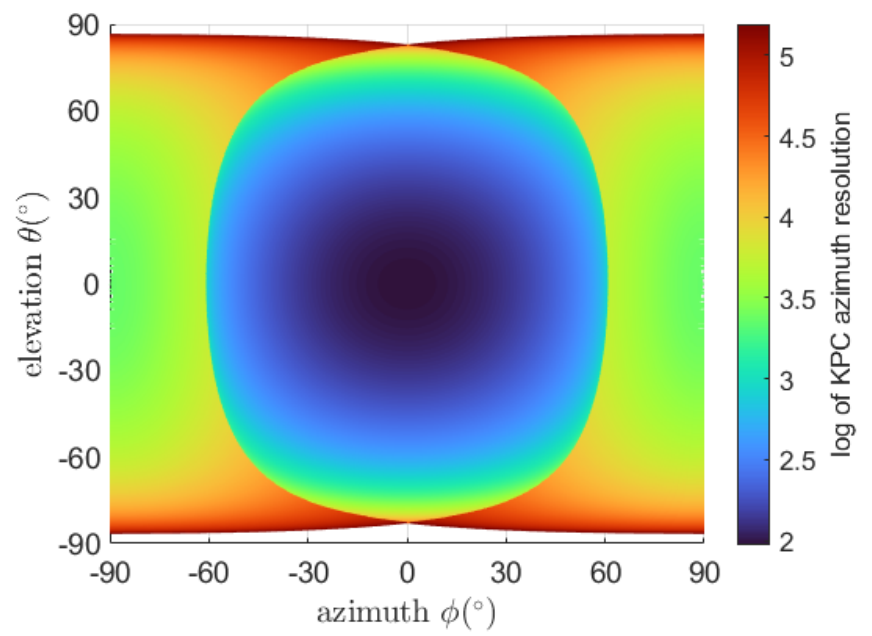}
  }
  \caption{Comparison of elevation and azimuth angular resolution between the proposed spherical DCAA and conventional UPA using KPC-based HBF. In each figure, the $x$-axis and $y$-axis represent the azimuth and elevation desired directions, respectively, while the color or magnitude indicates the angular resolution value in the corresponding direction. The resolution values, e.g., the color bar, are displayed on a logarithmic scale to highlight trends across directions and to illustrate the differences between spherical DCAA and conventional UPA using KPC-based HBF. (a) elevation resolution of spherical DCAA, (b) elevation resolution of conventional UPA using KPC-based HBF, (c) azimuth resolution of spherical DCAA , (d) azimuth resolution of conventional UPA using KPC-based HBF.    }
  \label{fig:resolution_compare}
\end{figure*}

Angular resolution is a key metric for evaluating the communication and sensing ability of an antenna array. It quantifies the array's ability to distinguish between two closely spaced incoming signals. The resolution depends on the beam pattern $ r(\phi, \theta; \phi', \theta')$, which is a function of both the \textbf{desired direction} $ (\phi', \theta') $, i.e., where the main lobe is pointed, and the \textbf{observation direction} $ (\phi, \theta) $, i.e., the angular variable at which the array response is evaluated. According to Lemma \ref{lemma:rotate_peak} and \ref{lemma:KPC peak direction}, for spherical DCAA the \textbf{desired direction} satisfies $(\phi',\theta')=(\eta,\vartheta)$, and for conventional UPA using KPC-based HBF we have $(\phi',\theta')=(\arcsin\frac{\sin\phi_p}{\cos\theta_q},\theta_q)$. For a given desired direction, the resolution is defined via the beam width of the main lobe along the elevation and azimuth axes, measured between the first nulls on either side of the peak. Specifically, for the elevation angular resolution, we fix the azimuth observation direction $\phi$ to $\phi'$, and varying the elevation observation direction $\theta$ around the desired direction $\theta'$, to find the first null points on both sides of it. A formal definition is given below.
\begin{definition}\label{def:resolution}
     The azimuth and elevation angular resolution $\gamma_h(\phi',\theta')$ and $ \gamma_v(\phi',\theta')$ are functions of the desired signal direction $(\phi',\theta')$, which is defined as half of the main lobe beam width $\Delta_\phi(\phi',\theta'),\Delta_\theta(\phi',\theta')$ of the beam pattern $r(\phi,\theta;\phi',\theta')$, i.e., 
\begin{equation}\label{eq:angular resolution vertical}
\begin{aligned}
    &\gamma_v(\phi',\theta')=\frac{1}{2}\Delta_\theta(\phi',\theta')=\frac{1}{2}|\theta_1-\theta_2|,\\
    &\gamma_h(\phi',\theta')=\frac{1}{2}\Delta_\phi(\phi',\theta')=\frac{1}{2}|\phi_1-\phi_2|,
    \end{aligned}
\end{equation}
\end{definition}
where $\theta_1,\theta_2$ and $\phi_1,\phi_2$ are the pairs of the closest nulls of $r(\phi,\theta;\phi',\theta')$ on both side of desired direction $\theta'$ and $\phi'$, as shown in Fig. \ref{fig:resolution_def}
\begin{figure}[!t] 
        \centering \includegraphics[width=0.65\columnwidth]{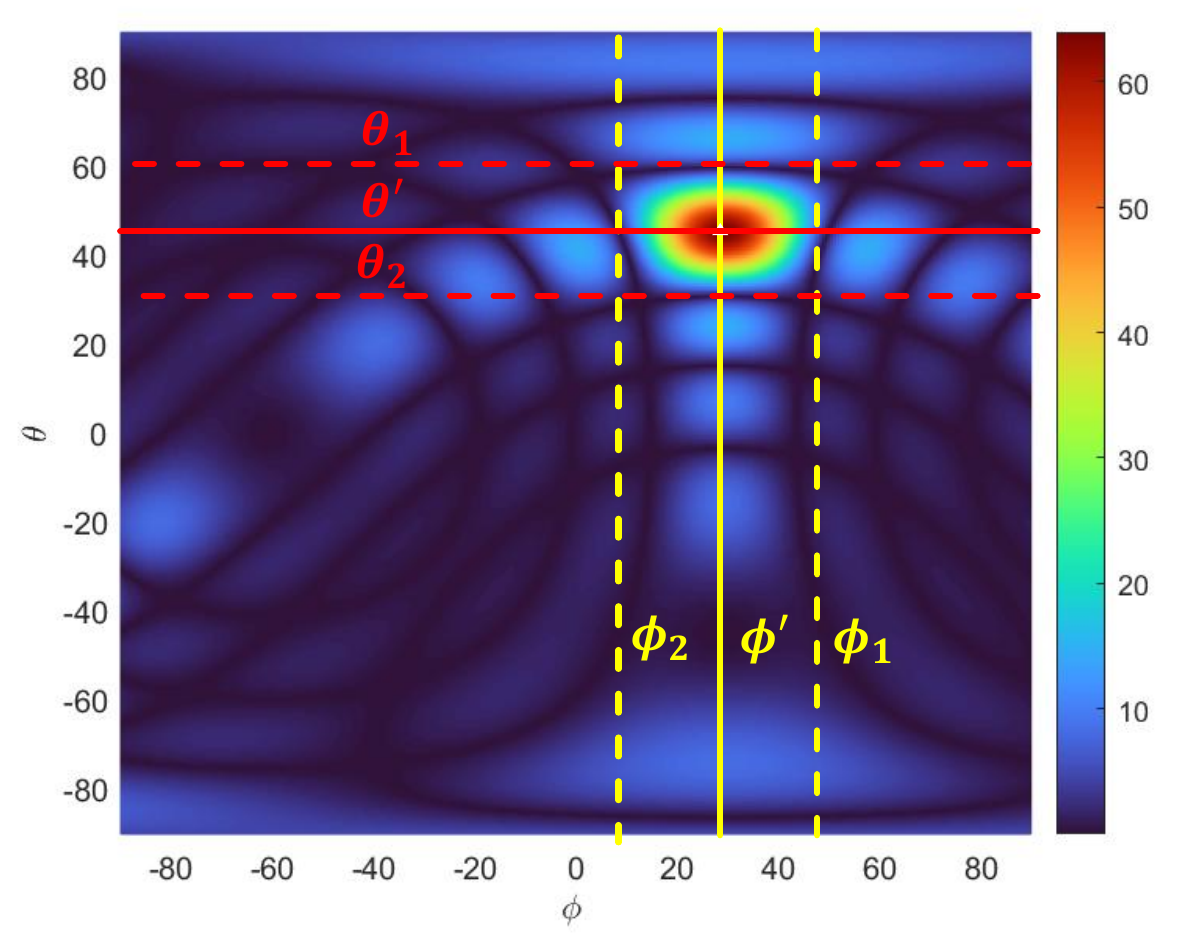}
        \caption{\label{fig:resolution_def}The illustration of $\theta_1,\theta_2$ and $\phi_1,\phi_2$}
\end{figure}

\begin{equation}\label{eq:ltheta}
\begin{aligned}
    \theta_1 = \mathop{\min}\limits_{\theta}\ \Big\{\theta>\theta'\big|r(\phi=\phi',\theta;\phi',\theta')=0\Big\},\\
    \theta_2 = \mathop{\max}\limits_{\theta}\ \Big\{\theta<\theta'\big|r(\phi=\phi',\theta;\phi',\theta')=0\Big\}.
    \end{aligned}
\end{equation}

\begin{equation}\label{eq:lvarphi}
\begin{aligned}
    \phi_1 = \mathop{\min}\limits_{\phi}\ \Big\{\phi>\phi'\big|r(\phi,\theta=\theta';\phi',\theta')=0\Big\},\\
    \phi_2 = \mathop{\max}\limits_{\phi}\ \Big\{\phi<\phi'\big|r(\phi,\theta=\theta';\phi',\theta')=0\Big\}.
        \end{aligned}
\end{equation}

By comparing Fig. \ref{fig:beam1} and \ref{fig:KPC_beam}, we can find that for the proposed spherical DCAA, the elevation angular resolution (or the elevation main lobe beamwidth) is uniform, and for the same desired elevation angle $\theta'$, it has uniform azimuth angular resolution (or the azimuth main lobe beamwidth) with respect to the desired azimuth angle $\phi'$. But in conventional UPA using KPC-based HBF, the azimuth angle resolution ability degrades as desired azimuth angle $\phi'$ gets larger. Besides, both spherical DCAA and conventional UPA using KPC-based HBF suffer from decreasing azimuth angle resolution ability as desired elevation angle $\theta'$ increases, as shown in the following theorem.

\begin{theorem}\label{the:resolution}
For any desired direction $(\phi',\theta')$ satisfies $\phi'>0,\theta'>0$, the elevation and azimuth angular resolution of spherical DCAA, denoted by $\gamma_v^{\mathrm{DCAA}}(\phi',\theta')$ and $\gamma_h^{\mathrm{DCAA}}(\phi',\theta')$ respectively, are given by:
\begin{equation}
\begin{aligned}
&\gamma_v^{\mathrm{DCAA}}(\phi',\theta')=\arcsin\frac{2}{M}, \\
&\gamma_h^{\mathrm{DCAA}}(\phi',\theta')=\arcsin\frac{2}{M\cos\theta'}, 
\end{aligned}
\end{equation}
and that of conventional UPA using KPC-based HBF, denoted by $\gamma_v^{\mathrm{UPA}}(\phi',\theta')$ and $\gamma_h^{\mathrm{UPA}}(\phi',\theta')$ respectively, are given by:

\begin{equation}
\begin{aligned}
&\gamma_v^{\mathrm{UPA}}(\phi',\theta')=\frac{1}{2}\big(\min\Theta^+-\max\Theta^-\big),\\
&\gamma_h^{\mathrm{UPA}}(\phi',\theta')=\frac{1}{2}\big(\min\Phi^+-\max\Phi^-\big),
\end{aligned}
\end{equation}
where $\Theta^+,\Theta^-,\Phi^+,\Phi^-$ are defined in Appendix \ref{app:resolution}.
\end{theorem}

\begin{IEEEproof}
Please refer to Appendix \ref{app:resolution}.
\end{IEEEproof}

The comparison of angular resolution is shown in Fig. \ref{fig:resolution_compare}. The blank area in Fig.\ref{fig:RAA_horizontal_resolution_top}-\ref{fig:KPC_horizontal_resolution_top} when $\theta$ is near $90^\circ$ means that the null-to-null angular resolution defined in \ref{def:resolution} does not exist in such areas. It shows that the elevation angular resolution of spherical DCAA is uniform in both dimensions, while the azimuth angular resolution is uniform in the azimuth dimension, but increases in the elevation dimension. The elevation angular resolution of conventional UPA using KPC-based HBF does not exhibit a monotonically increasing or decreasing trend in either dimension, and the azimuth angular resolution increases in both dimensions. In general, we can conclude that the proposed spherical DCAA has not only wider angular resolution coverage than conventional UPA using KPC-based HBF, but also better angular resolution ability.

\section{Sensing Algorithm for Spherical DCAA-based ISAC}\label{sec:algorithm}

\subsection{Spherical DCAA-based ISAC}
By substituting \eqref{eq:channel} into \eqref{eq:rx antenna 2}, we have

\begin{equation}\label{eq:receivey}
    \mathbf{y}=\mathbf{S}\mathbf{\tilde{y}}=\mathbf{S}\sum_{l=0}^{L}\alpha_l\mathbf{r}(\phi_l,\theta_l)x_s+\mathbf{S}\mathbf{z},
\end{equation}
where $\mathbf{z}=\left[ \sum_{m=0}^{M^2-1}z_{m,n} \right]_{n\in \mathcal{N}}\in\mathbb{C}^{N\times 1}$, and $z_{m,n}$ denotes the noise from the $m$th element of the $n$th sUPA. Suppose $z_{m,n}$ follows i.i.d. CSCG distribution with variance $\sigma^2$ at each time sample, i.e., $z_{m,n}\sim\mathcal{CN}(0,\sigma^2)$. Then we can get $\mathbf{z}\sim\mathcal{CN}(0,M^2\sigma^2\mathbf{I}_N)$. 

\subsubsection{Communication Model}
From the communication perspective, the received signal \eqref{eq:receivey} at the Rx can be written as
\begin{equation}
    \mathbf{y}_{\mathrm{c}}=\mathbf{h}_{\mathrm{c}}x_s+\mathbf{n}_{\mathrm{c}},
\end{equation}
where $\mathbf{h}_{\mathrm{c}}=\mathbf{S}\sum_{l=0}^{L}\alpha_l\mathbf{r}(\phi_l,\theta_l)\in\mathbb{C}^{N_{\mathrm{RF}}\times 1}$ is the equivalent communication channel, and $\mathbf{n}_{\mathrm{c}}=\mathbf{S}\mathbf{z}$ denotes the denotes the AWGN vector.
Accordingly, the uplink communication rate is given by:
\begin{equation}
    \mathcal{R}=\log_2\Big(1+\frac{\|\mathbf{h}_\mathrm{c}\|^2P_t}{M^2\sigma^2}\Big).
\end{equation}

\subsubsection{Sensing Model}
For sensing, the signal \eqref{eq:receivey} can be reformulated as $\mathbf{y}_{\mathrm{s}}$
\begin{equation}\label{eq:sensing data}
\begin{aligned}
\mathbf{y}_{\mathrm{s}}=\sum_{l=0}^{L}\bar\alpha_l\mathbf{h}_{\mathrm{s}}(\phi_l,\theta_l)+\mathbf{n}_{\mathrm{s}},
\end{aligned}
\end{equation}
where $\bar\alpha_l=\alpha_lx_s$, $\mathbf{h}_{\mathrm{s}}(\phi_l,\theta_l)=\mathbf{S}\mathbf{r}(\phi_l,\theta_l)\in\mathbb{C}^{N_{\mathrm{RF}}\times 1}$ is the equivalent steering vector, and $\mathbf{n}_{\mathrm{s}}=\mathbf{S}\mathbf{z}\in\mathbb{C}^{N_{\mathrm{RF}}\times 1}$ denotes the noise vector. Based on \eqref{eq:sensing data}, we try to obtain $\phi_l$ and $\theta_l$. 


\subsection{Sensing Algorithm Design}

By concatenating the signal $\mathbf{y}$ from $K$ different time slots, we obtain the snapshot matrix $\mathbf{Y} \in \mathbb{C}^{N_{\mathrm{RF}} \times K}$ 
\begin{equation}\label{eq:snapshot}
    \mathbf{Y}(:,k)=\mathbf{y}[k],
\end{equation}
where $k=1,2,...,K$ is the index for the $k$th time slot. It can be reformulated as
\begin{equation}
\mathbf{Y} = \mathbf{H} \mathbf{X},
\end{equation}
where $\mathbf{H} = [\mathbf{h}_{\mathrm{s}}(\phi_0,\theta_0), \ldots, \mathbf{h}_{\mathrm{s}}(\phi_L,\theta_L)]\in\mathbb{C}^{N_{\mathrm{RF}}\times L}$ is the array manifold matrix and $\mathbf{X}\in\mathbb{C}^{L\times K}$ contains the complex amplitudes. The AoA estimation problem aims to estimate the angle pairs $\mathbb{S} = \{(\phi_0,\theta_0), \ldots, (\phi_L,\theta_L)\}$ from $\mathbf{Y}$. Since the steer vector of spherical DCAA lacks rotational invariance, we employ the MUSIC algorithm. The covariance matrix $\mathbf{C}_{\mathbf{Y}}\in\mathbb{C}^{N_{\mathrm{RF}}\times N_{\mathrm{RF}}}$ of $\mathbf{Y}$ is computed as
\begin{equation}
\mathbf{C}_{\mathbf{Y}} = \frac{1}{K} \mathbf{Y} \mathbf{Y}^{\mathrm{H}},
\end{equation}
and its eigenvalue decomposition (EVD) $\mathbf{C}_{\mathbf{Y}}=\mathbf{E}_{\mathrm{s}}\mathbf{\Sigma}_{\mathrm{s}}\mathbf{E}_{\mathrm{s}}^{\mathrm{H}}+\mathbf{E}_{\mathrm{n}}\mathbf{\Sigma}_{\mathrm{n}}\mathbf{E}_{\mathrm{n}}^{\mathrm{H}}$ yields the signal subspace matrix $\mathbf{E}_{\mathrm{s}}\in\mathbb{C}^{N_{\mathrm{RF}}\times (L+1)}$, noise subspace matrix $\mathbf{E}_{\mathrm{n}}\in\mathbb{C}^{N_{\mathrm{RF}}\times (N_{\mathrm{RF}}-L-1)}$, and diagonal matrices $\mathbf{\Sigma}_{\mathrm{s}}\in\mathbb{C}^{(L+1)\times (L+1)}$ and $\mathbf{\Sigma}_{\mathrm{n}}\in\mathbb{C}^{(N_{\mathrm{RF}}-L-1)\times (N_{\mathrm{RF}}-L-1)}$. The MUSIC spatial spectrum is then formulated as
\begin{equation}
P_{\mathrm{MUSIC}}(\phi,\theta) = \frac{1}{\mathbf{h}_{\mathrm{s}}^{\mathrm{H}}(\phi,\theta) \mathbf{E}_{\mathrm{n}} \mathbf{E}_{\mathrm{n}}^{\mathrm{H}} \mathbf{h}_{\mathrm{s}}(\phi,\theta)}.
\end{equation}

The peaks of this spectrum correspond to the estimated AoAs, denoted by $\hat{\mathbb{S}} = \{(\hat{\phi}_0,\hat{\theta}_0), \ldots, (\hat{\phi}_{L_p},\hat{\theta}_{L_p})\}$, where $L_p \leq L+1$. The procedure is summarized in Algorithm~\ref{alg:algorithm}.

\begin{algorithm}[htbp]
\caption{Proposed Algorithm for ISAC with spherical DCAA}
\label{alg:algorithm}
\renewcommand{\algorithmicrequire}{\textbf{Input:}}
\renewcommand{\algorithmicensure}{\textbf{Output:}}
\begin{algorithmic}[1]
\REQUIRE Signal $\mathbf{y}_{\mathrm{s}}$ in \eqref{eq:sensing data}
\ENSURE $\phi_{l}$, $\theta_{l}$, $\forall l$
\STATE Obtain $\mathbf{Y}$ from $\mathbf{y}_{\mathrm{s}}$ as in \eqref{eq:snapshot}
\STATE Compute covariance matrix $\mathbf{C}_{\mathbf{Y}}$
\STATE Perform EVD of $\mathbf{C}_{\mathbf{Y}}$ to obtain noise subspace $\mathbf{E}_{\mathrm{n}}$
\FOR{$\theta \in [-\theta_\mathrm{max}, \theta_\mathrm{max}]$}
\FOR{$\phi \in [-\phi_\mathrm{max}, \phi_\mathrm{max}]$}
\STATE Calculate steering vector $\mathbf{h}_{\mathrm{s}}(\phi,\theta)$
\STATE Evaluate $P_{\mathrm{MUSIC}}(\phi,\theta)$
\ENDFOR
\ENDFOR
\STATE Find peaks of $P_{\mathrm{MUSIC}}(\phi,\theta)$ to obtain $\hat{\mathbb{S}}$
\RETURN $\phi_l, \theta_l$ for $l = 0,1,\ldots,L_p$
\end{algorithmic}
\end{algorithm}

\section{Simulation Results}\label{sec:simulation}

\begin{figure}[t]
  \centering
  \subfloat[Spherical DCAA \label{fig:}]{
    \includegraphics[width=0.4\linewidth]{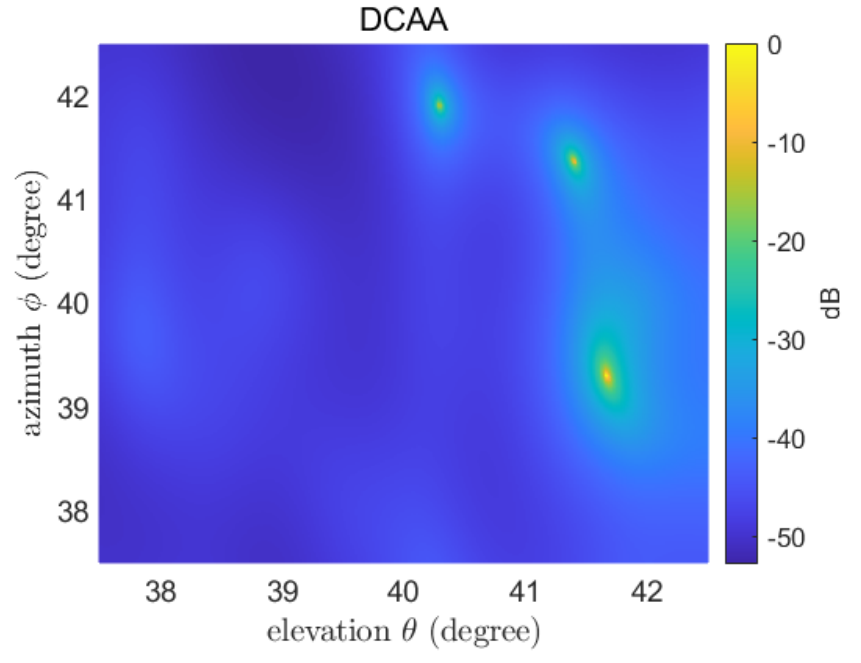}
  }
  \subfloat[Spherical DCAA \label{fig:}]{
    \includegraphics[width=0.4\linewidth]{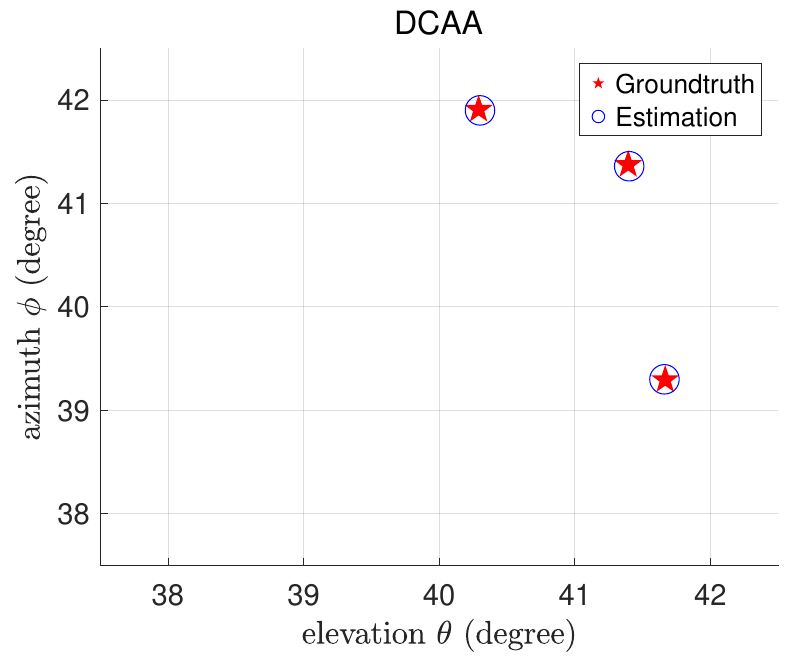}
  }
  \\
  \subfloat[UPA with conventional KPC\label{fig:}]{
    \includegraphics[width=0.4\linewidth]{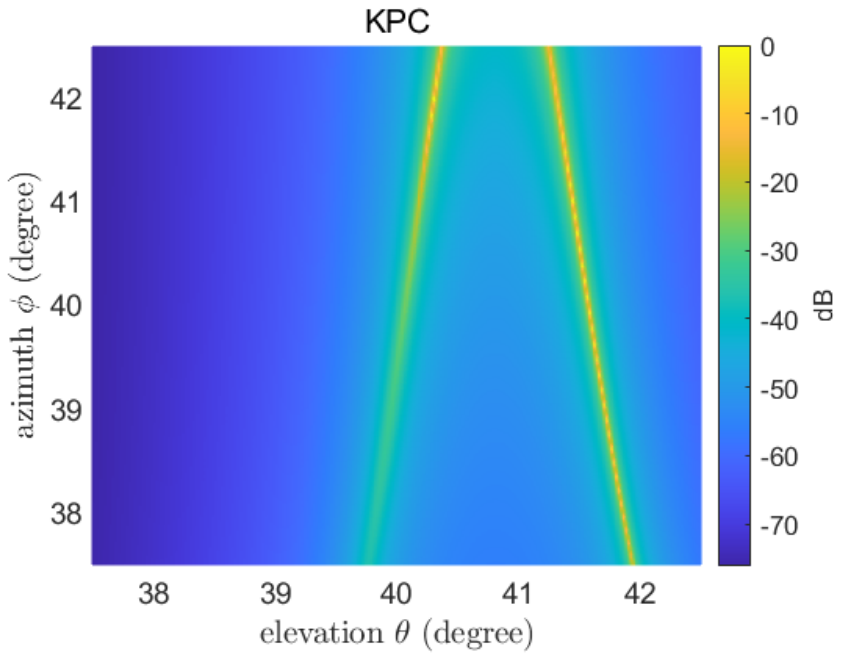}
  }
  \subfloat[UPA with conventional KPC\label{fig:}]{
    \includegraphics[width=0.4\linewidth]{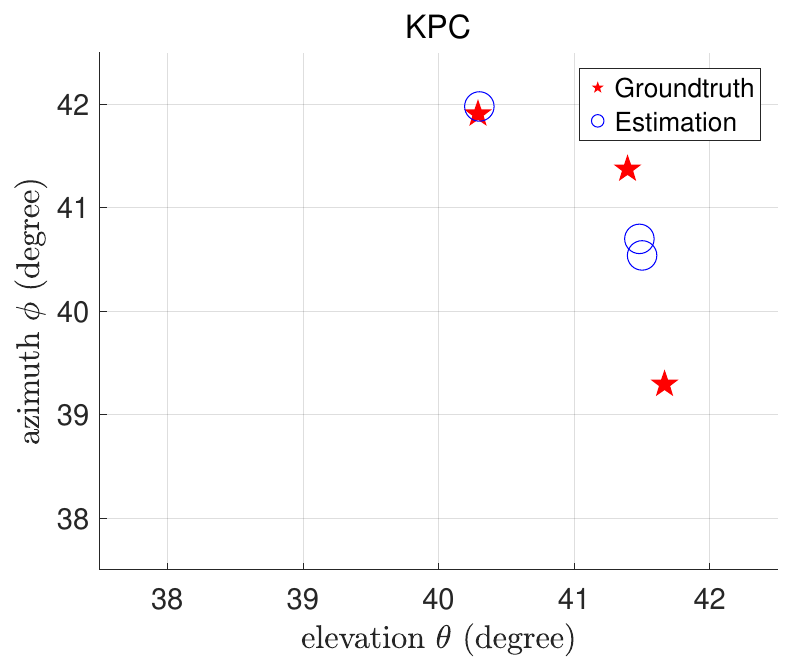}
  }
  
  \caption{2D angle estimation results when the swarm is centered at $(40^\circ,40^\circ)$. Figure (a)(c) are the MUSIC spectrum. Figure (b)(d) are the groundtruth and estimation results.}
  \label{fig:2d angle result}
\end{figure}

In this section, we provide numerical results to verify the performance of spherical DCAA in ISAC system. For spherical DCAA, we set the maximum orientation of sUPAs and antenna elements in each sUPA as $\vartheta_{\max}=\eta_{\max}=\pi/2,M^2=256$ at Rx, respectively. Therefore, we can compute $N_\vartheta=12,N_{\eta}(\vartheta)=\left \lfloor \frac{\pi}{2}/\arcsin\frac{2}{16\cos\vartheta} \right \rfloor$, and the total number of sUPA is $N=397$. 
For KPC-based UPA, the number of elevation and azimuth codewords is $N_v=17$ and $N_h=17$, and the $(p,q)$th codeword’s angle is $(\phi_p,\theta_q)=(\arcsin\frac{2p}{16},\arcsin\frac{2q}{16}),-8\le p,q\le 8$. The number of RF chains is set to $N_{\mathrm{RF}}=8$. The carrier frequency is set as $f_c=39$ GHz. The radius of the DCAA sphere is $R=0.2705$ m.
In addition, the radiation patterns of antenna elements follow the 3GPP antenna model, which is expressed in dB as \cite{3GPPchannel}
\begin{equation}
\begin{aligned}
    &G_{\mathrm{dB}}(\phi,\theta)\\
    &=-\min\big\{ -\big(A_{\mathrm{dB}}(\phi=0,\theta) + A_{\mathrm{dB}}(\phi,\theta=0)\big),A_{\mathrm{max}} \big\},\\
{\text{with }}
&A_{\mathrm{dB}}(\phi=0,\theta)=A_0^{\mathrm{dB}}-\min\Big\{ 12\Big(\frac{\theta}{\theta_{\mathrm{3dB}}}\Big)^2, \mathrm{SLA_V} \Big\},\\
&A_{\mathrm{dB}}(\phi,\theta=0)=A_0^{\mathrm{dB}}-\min\Big\{ 12\Big(\frac{\phi}{\phi_{\mathrm{3dB}}}\Big)^2, A_{\mathrm{max}} \Big\},
\end{aligned}
\end{equation}
among which $\theta_{\mathrm{3dB}}$ and $\phi_{\mathrm{3dB}}$ account for the elevation and azimuth 3dB beamwidth, $A_0^{\mathrm{dB}}$ is the peak antenna gain in dB, $A_{\mathrm{max}}=\mathrm{SLA_V}=30\mathrm{dB}$. For conventional UPA, $\theta_{\mathrm{3dB}}^{\mathrm{KPC}},\phi_{\mathrm{3dB}}^{\mathrm{KPC}}$ is set to $\pi$ to cover the entire direction, while for spherical DCAA, $\theta_{\mathrm{3dB}}^{\mathrm{DCAA}},\phi_{\mathrm{3dB}}^{\mathrm{DCAA}}$ is set to $0.3\pi$, as each sUPA is only responsible for a narrower angular region. Besides, $A_0^{\mathrm{dB}}$ in conventional UPA and spherical DCAA are set to 0dB and 12.79dB to guarantee the same total power gain for all directions. 



A UAV swarm is simulated by three low-altitude UAV targets with uniformly distributed elevation and azimuth angles in a confined area with $5^\circ \times 5^\circ$ angular region. Note that due to the close proximity of UAV swarms, adjacent targets may not be perfectly discriminated, so the criterion termed as \textit{average missed targets} is incorporated to evaluate the performance, which is defined as
\begin{equation}
    \varepsilon  =\frac{1}{Q}\sum_{i=1}^{Q}\Big( \mathrm{card}\big(\mathbb{S}^i \big)-\mathrm{card}\big(\hat{\mathbb{S}}^i \big) \Big),
\end{equation}
where $Q$ denotes the total number of testing rounds, $\mathbb{S}^i$ and $\hat{\mathbb{S}}^i$ denote the real and estimation angle set of $i$th test round, respectively. We assume that $\mathrm{card}\big(\mathbb{S}^i \big)\ge\mathrm{card}\big(\hat{\mathbb{S}}^i \big)$, which holds when the noise is rather small.

The 2D angle estimation results for the UAV swarm are presented in Fig. \ref{fig:2d angle result}. When the swarm is centered at $(40^\circ, 40^\circ)$, the proposed spherical DCAA method successfully detects all three targets and estimates their angles $(\theta,\phi)$ with high precision. In contrast, the conventional UPA fails to distinguish these targets due to severe spectral aliasing, leading to significant estimation errors. This performance degradation is primarily attributed to the inferior angular resolution, i.e., wider main lobe beamwidth of the conventional UPA.


Fig. \ref{fig:missing compare} compares the average number of missed targets between spherical DCAA and conventional UPA. Overall, spherical DCAA demonstrates superior detection ability. For a given elevation UAV swarm center angle $\theta$, its average missed target count remains nearly constant across different azimuth angles $\phi$, due to its invariant angular resolution in the azimuth domain. However, this count increases with $\theta$ because the azimuth angular resolution degrades at higher elevations, affecting the detection ability. Conversely, for conventional UPA, the number of missed targets generally increases with both $\theta$ and $\phi$, consistent with a progressive loss of angular resolution. An exception occurs at $\theta = 80^\circ$, where its elevation angular resolution decreases as $\phi$ increases. This positive effect outweighs the increase in azimuth resolution, leading to better detection performance at larger $\phi$.

\begin{figure}[t]
  \centering
  \subfloat[\label{fig:}]{
    \includegraphics[width=0.45\linewidth]{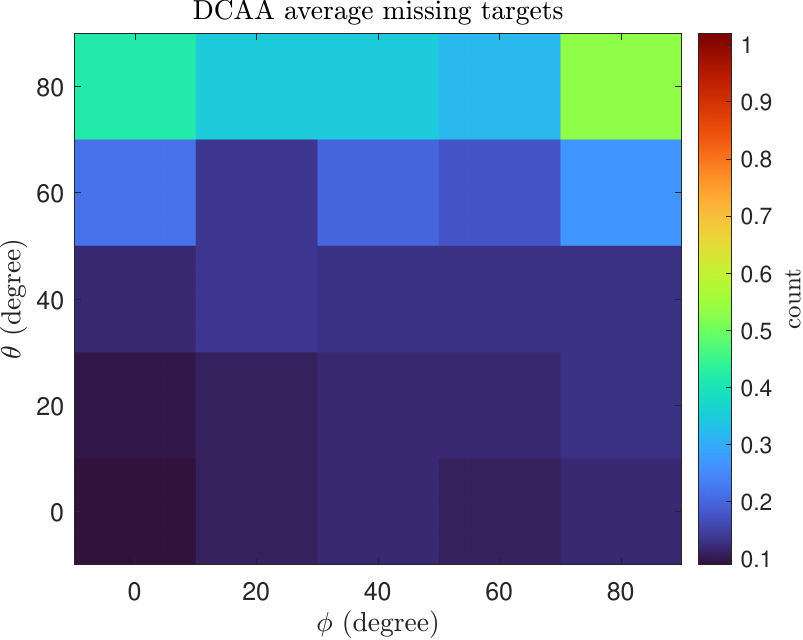}
  }
  \subfloat[\label{fig:}]{
    \includegraphics[width=0.45\linewidth]{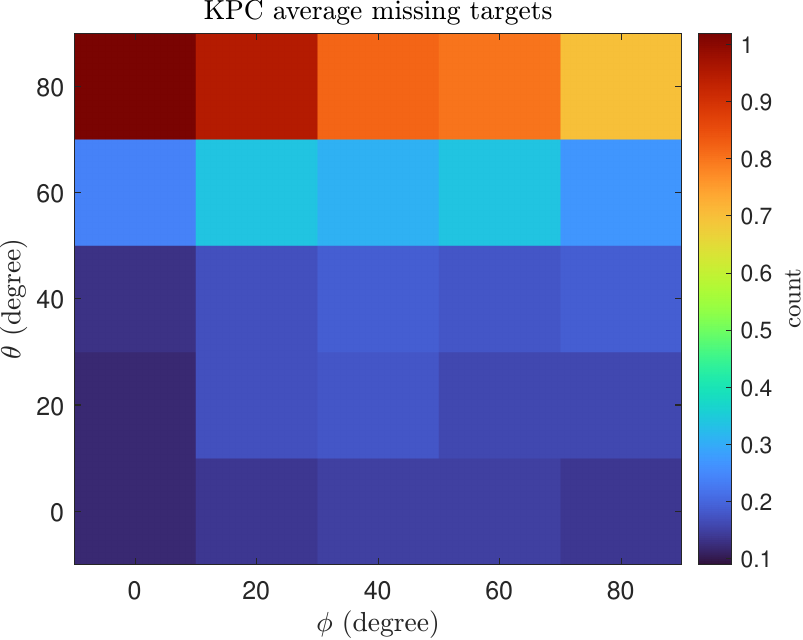}
  }
  
  \caption{Comparison of average missing targets between spherical DCAA and conventional UPA using KPC-based HBF. In each figure, the $x$-axis and $y$-axis represent the azimuth and elevation UAV swarm center, respectively. (a) spherical DCAA, (b) conventional UPA using KPC-based HBF.}
  \label{fig:missing compare}
\end{figure}

Fig. \ref{fig:angle compare} presents the RMSE comparison for angle $(\theta,\phi)$ estimation between spherical DCAA and conventional UPA using KPC-based HBF. Overall, spherical DCAA achieves higher estimation accuracy. For spherical DCAA, the elevation angle RMSE is lower than the azimuth angle RMSE. Both of them exhibit similar trends: the RMSE remains stable for a fixed $\theta$ across varying $\phi$, owing to a constant angular resolution in the azimuth domain. As $\theta$ increases, however, the azimuth angular resolution degrades, leading to a rise in the azimuth angle estimation RMSE. The elevation angle estimation RMSE also increases because the broadening azimuth main lobe interferes with the elevation main lobe, degrading estimation performance. For conventional UPA using KPC-based HBF, the elevation angle RMSE is similarly lower than the azimuth angle RMSE, and both generally increase with $\theta$ and $\phi$, consistent with its deteriorating angular resolution. 

\begin{figure*}[t]
  \centering
  \subfloat[\label{fig:DCAA_theta_RMSE}]{
    \includegraphics[width=0.23\linewidth]{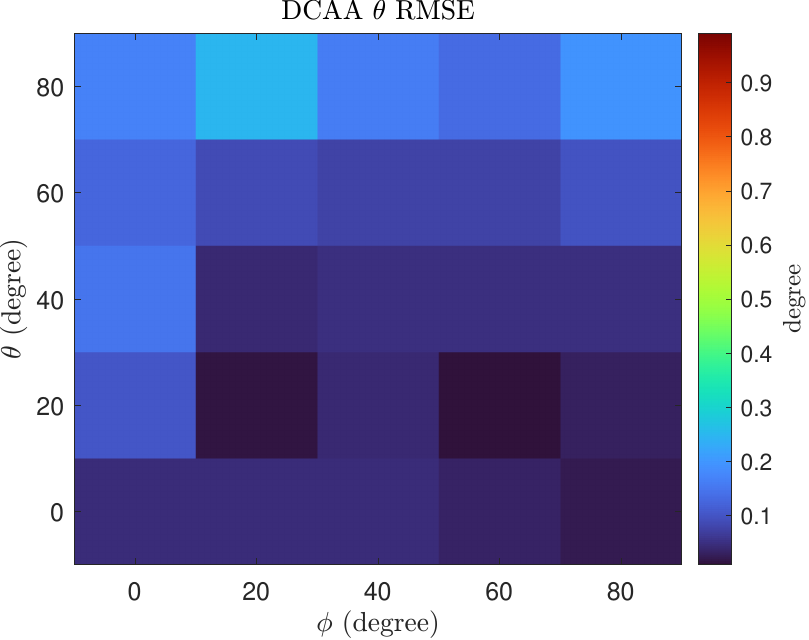}
  }
  \subfloat[\label{fig:KPC_theta_RMSE}]{
    \includegraphics[width=0.23\linewidth]{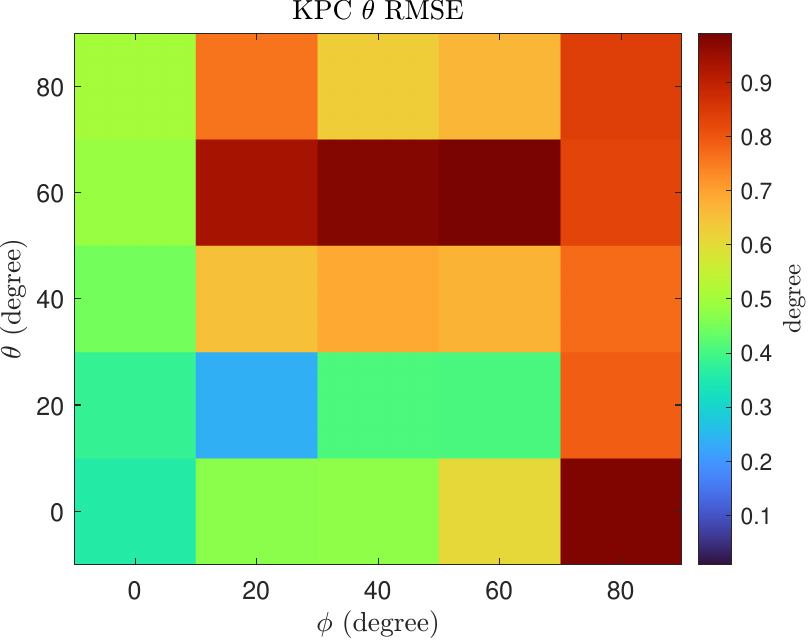}
  }
  \subfloat[\label{fig:DCAA_phi_RMSE}]{
    \includegraphics[width=0.23\linewidth]{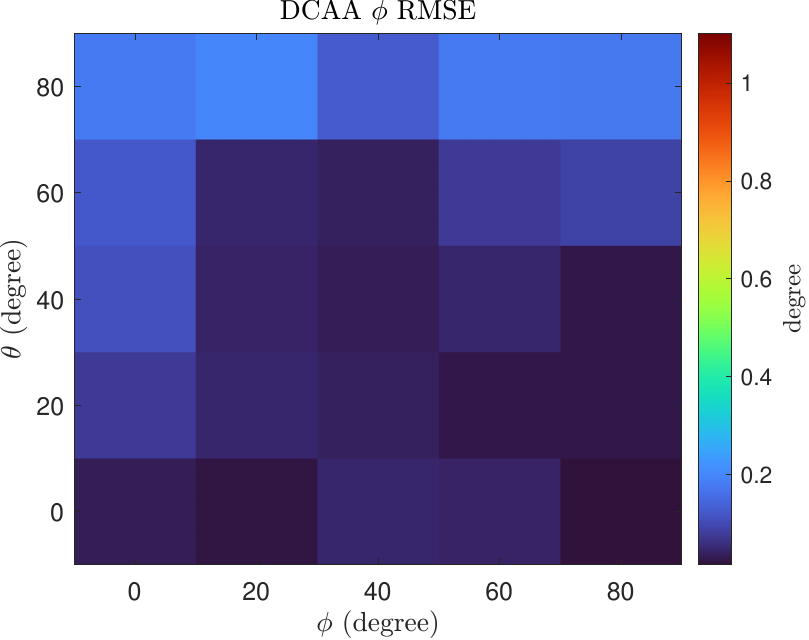}
  }
  \subfloat[\label{fig:DCAA_phi_RMSE}]{
    \includegraphics[width=0.23\linewidth]{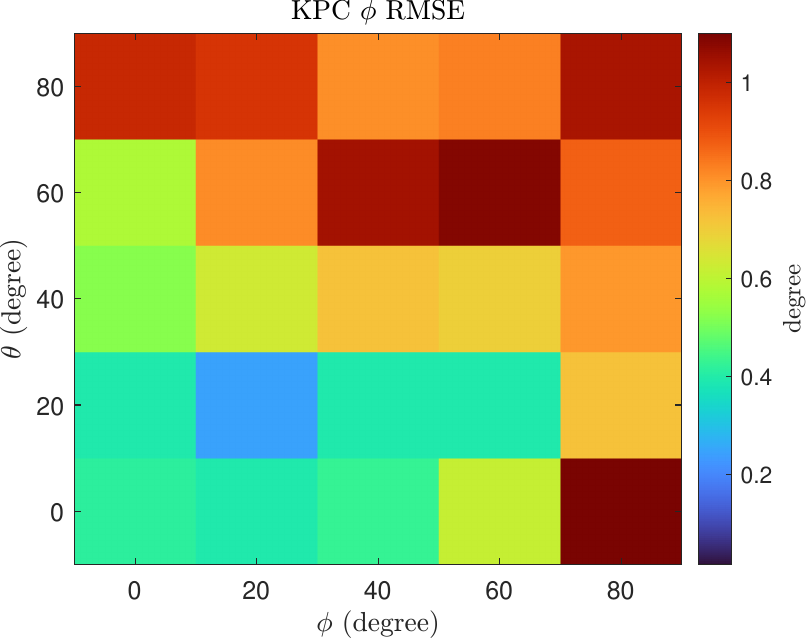}
  }
  \caption{Comparison of elevation and azimuth angular estimation RMSE between spherical DCAA and conventional UPA using KPC-based HBF. In each figure, the $x$-axis and $y$-axis represent the azimuth and elevation UAV swarm center, respectively. (a) elevation angle estimation of spherical DCAA, (b) elevation angle estimation of conventional UPA using KPC-based HBF, (c) azimuth angle estimation of spherical DCAA, (d) azimuth angle estimation of conventional UPA using KPC-based HBF. }
  \label{fig:angle compare}
\end{figure*}

  

  

Fig. \ref{fig:rate} compares spectrum efficiency between the two types of array with and without directional antenna elements. It can be seen that when the directional antenna elements are employed, spherical DCAA achieves significantly higher spectrum efficiency than that of conventional UPA using KPC-based HBF across different transmit SNR levels. This performance improvement is attributed to the enhanced beamforming gain of the spherical DCAA. In contrast, when isotropic antennas are used, conventional UPA using KPC-based HBF slightly outperforms spherical DCAA, as spherical DCAA’s finer angular resolution makes it more challenging to align exactly with the signal directions. In this condition, the signal strength of spherical DCAA degrades faster than conventional UPA.

\begin{figure}[htbp] 
        \centering \includegraphics[width=0.65\columnwidth]{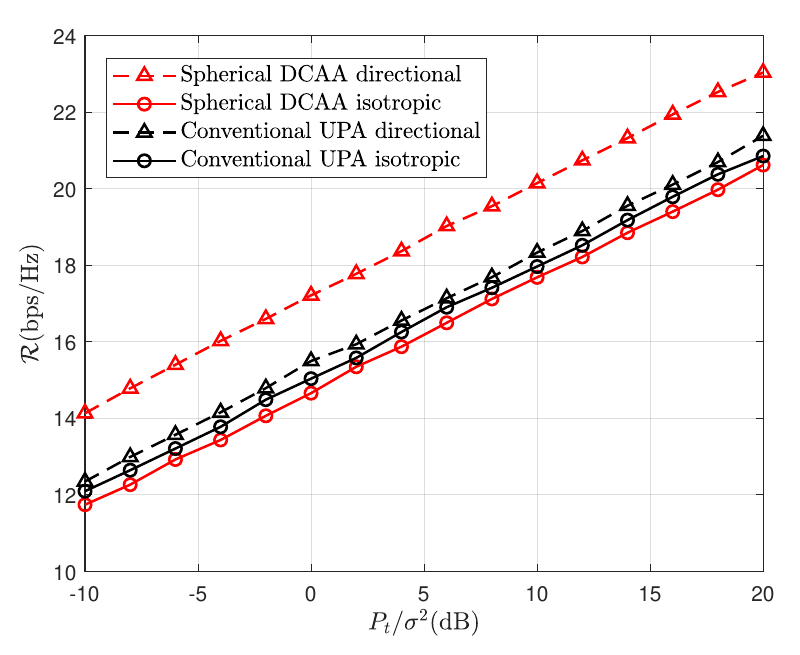}
        \caption{\label{fig:rate}Comparison of spectrum efficiency between spherical DCAA and conventional UPA using KPC-based HBF with and without directive antenna elements.}
\end{figure}

  

\section{Conclusion}\label{sec:conclusion}
In this paper, we proposed a novel spherical DCAA architecture for low-altitude UAV swarm ISAC. Compared to conventional arrays with hybrid analog/digital beamforming, the spherical DCAA offers significant hardware cost reduction, improved energy focusing, and superior angular resolution—with uniform elevation resolution and elevation-dependent azimuth resolution. A systematic spherical arrangement of sUPAs along with a selection network was designed, and an efficient sensing algorithm was developed for joint parameter estimation. Simulation results demonstrate that the spherical DCAA outperforms conventional arrays in angular resolution, target detection, and spectral efficiency, confirming its strong potential for future UAV swarm ISAC systems.

{\appendices

\section{Proof of Lemma \ref{lemma:sUPA beamwidth}}\label{app:sUPA beamwidth}

According to equation \eqref{eq:null 1}, for the main lobe beamwidth in the azimuth direction, $\theta$ is fixed as $\theta=\vartheta$ and the null points satisfy
\begin{equation}\label{eq:horizontal null 1}
\begin{aligned}
    &\phi_{\mathrm{null}}=\eta+\arcsin\Big(\frac{2p}{M\cos\vartheta}\Big),\\
{\text{or }} 
    &\phi_{\mathrm{null}}=\eta+\arccos\Big(1-\frac{2q}{M\cos\vartheta\sin\vartheta}\Big),
        \end{aligned}
\end{equation}
where $p=\pm1,q=\pm1$. Note that $\arcsin\big(\frac{2p}{M\cos\vartheta}\big)<\arccos\big(1-\frac{2p}{M\cos\vartheta}\big)<\arccos\big(1-\frac{2q}{M\cos\vartheta\sin\vartheta}\big)$,
the adjacent null points in the $\phi$ direction are 
\begin{equation}
\phi_{\mathrm{null}}=\eta+\arcsin\big(\frac{2p}{M\cos\vartheta}\big).
\end{equation}

By setting $p=\pm 1$, we have
\begin{equation}
    \Delta\phi=2\arcsin\big(\frac{2}{M\cos\vartheta}\big).
\end{equation}

Similarly, For the main lobe beam width in the elevation direction, $\phi$ is fixed as $\phi=\eta$ and the adjacent null points satisfy
\begin{equation}\label{eq:horizontal null 3}
    \theta_{\mathrm{null}}=\vartheta+\arcsin\Big(\frac{2q}{M}\Big),
\end{equation}
where $q=\pm1$, which yields
\begin{equation}
    \Delta\theta=2\arcsin\big(\frac{2}{M}\big).
\end{equation}

Therefore, the null to null beamwidth of the main lobe can be obtained as
\begin{equation}
    \textbf{BW}(\vartheta,\eta)=\bigg(2\arcsin\big(\frac{2}{M\cos\vartheta}\big),2\arcsin\big(\frac{2}{M}\big)\bigg).
\end{equation}

\section{Proof of Theorem \ref{the:angle_separation}}\label{app:minimum angle}
\begin{figure}[htbp] 
    \centering 
    \includegraphics[width=0.45\columnwidth]{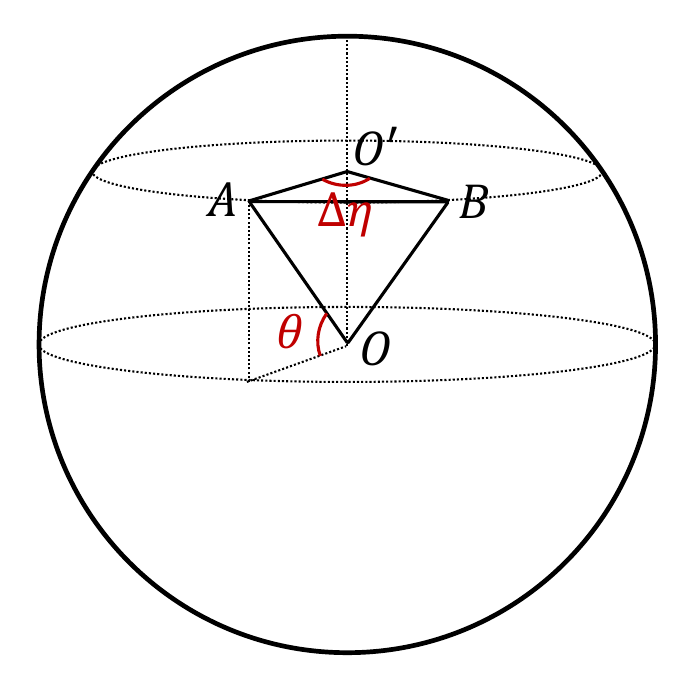}
    \caption{Demonstration of adjacent sUPAs with same elevation angle.}
    \label{fig:the_angle_separation}
\end{figure}

The direction of a sUPA is given by
\begin{equation}
    (\eta_{p,q},\vartheta_q) = \left(p \arcsin\left(\frac{2}{M\cos(q\arcsin\frac{2}{M})}\right),\; q\arcsin\frac{2}{M}\right).
\end{equation}

Let two sUPAs $A$ and $B$ have directions $(\eta_a,\vartheta_a)$ and $(\eta_b,\vartheta_b)$, respectively.  
Regarded as points on a unit sphere centered at $O$, their Cartesian coordinates are
\begin{align}
    (x_a,y_a,z_a) &= (\cos\vartheta_a\cos\eta_a,\;\cos\vartheta_a\sin\eta_a,\;\sin\vartheta_a), \\
    (x_b,y_b,z_b) &= (\cos\vartheta_b\cos\eta_b,\;\cos\vartheta_b\sin\eta_b,\;\sin\vartheta_b).
\end{align}

Hence the cosine of the angle $\alpha^+$ between $\overrightarrow{OA}$ and $\overrightarrow{OB}$ is
\begin{equation}\label{eq:cos_alpha}
    \cos\alpha^+ = \overrightarrow{OA}\cdot\overrightarrow{OB}
    = \cos\vartheta_a\cos\vartheta_b\cos(\eta_a-\eta_b) + \sin\vartheta_a\sin\vartheta_b.
\end{equation}

\noindent\textbf{Case 1: Different elevation angles.}

If $\vartheta_a\neq\vartheta_b$, then from \eqref{eq:cos_alpha},
\begin{equation*}
    \cos\alpha^+ \le \cos(\vartheta_a-\vartheta_b),
\end{equation*}
so that $|\alpha^+|\ge|\vartheta_a-\vartheta_b|\ge\arcsin\frac{2}{M}$.

\noindent\textbf{Case 2: Same elevation angle.}

Let $\vartheta_a=\vartheta_b=\vartheta$.  
Denote $c=\cos\vartheta$, $s=\sin\vartheta$.  
From \eqref{eq:cos_alpha} and the expression of $\eta_{p,q}$, we have
\begin{equation}\label{eq:cos_beta}
    \cos\alpha^+ = c^2\cos\!\left(\arcsin\frac{2}{Mc}\right) + s^2
    = c^2\sqrt{1-\frac{4}{M^2c^2}} + s^2 \triangleq f(\vartheta),
\end{equation}
where $0<\vartheta<\arccos\frac{2}{M}$ (so that $c>\frac{2}{M}$).  
Define $t = Mc\in(2,M)$ and set
\begin{equation*}
    w(t) = t - \frac{t^2-2}{\sqrt{t^2-4}}.
\end{equation*}

Differentiating $f(\vartheta)$ gives
\begin{equation}\label{eq:f_prime}
    f'(\vartheta) = \frac{2s}{M}\,w(t).
\end{equation}

Since $s>0$ for $\vartheta\in(0,\arccos\frac{2}{M})$, the sign of $f'(\vartheta)$ is that of $w(t)$.  
The function $w(t)$ is continuous on $(2,M)$ and has no zero there:  
if $w(t)=0$, then $t\sqrt{t^2-4}=t^2-2$, squaring which leads to the contradiction $0=4$.  
Evaluating at $t=3$ gives $w(3)=3-7/\sqrt{5}\approx-0.13<0$; hence $w(t)<0$ for all $t\in(2,M)$.  
Consequently $f'(\vartheta)<0$, so $f(\vartheta)$ is strictly decreasing on $[0,\arccos\frac{2}{M})$.  
Its maximum is therefore attained at $\vartheta=0$:
\begin{equation*}
    f(\vartheta) \le f(0) = \sqrt{1-\frac{4}{M^2}}.
\end{equation*}

From \eqref{eq:cos_beta} we obtain $\cos\alpha^+ \le \sqrt{1-\frac{4}{M^2}}$, which implies
\begin{equation*}
    \alpha^+ \ge \arccos\sqrt{1-\frac{4}{M^2}} = \arcsin\frac{2}{M}.
\end{equation*}

\noindent\textbf{Combining both cases}, we conclude that for any two sUPAs,
\begin{equation*}
    \boxed{\;\alpha^+ \ge \arcsin\frac{2}{M}\;.}
\end{equation*}

\section{Proof of Theorem \ref{the:resolution}}\label{app:resolution}
For spherical DCAA, according to \eqref{eq:sUPA response with rotate} and \eqref{eq:angular resolution vertical}-\eqref{eq:lvarphi}, by letting $\phi=\phi'$, we can get $\sin(\theta-\theta')=\pm2/M$, which yields $\theta_1=\theta'-\arcsin\frac{2}{M},\theta_2=\theta'+\arcsin\frac{2}{M}$. Then, the elevation angular resolution can be obtained as $\gamma_v^{\mathrm{DCAA}}(\phi',\theta')=\frac{1}{2}\Delta_\theta(\phi',\theta')=\frac{1}{2}|\theta_1-\theta_2|=\arcsin\frac{2}{M}$. Similarly, by letting $\theta=\theta'$, we will have $\cos\theta'\sin(\phi-\phi')=\pm2/M$ or $\sin\theta'\cos\theta'(1-\cos(\phi-\phi'))=\pm2/M$. Based on previous result \eqref{eq:horizontal null 1}-\eqref{eq:horizontal null 3}, $\phi_1=\phi'-\arcsin(\frac{2}{M\cos\theta'}),\phi_2=\phi'+\arcsin(\frac{2}{M\cos\theta'})$, and the azimuth angular resolution can be obtained as $\gamma_h^{\mathrm{DCAA}}(\phi',\theta')=\frac{1}{2}\Delta_\phi(\phi',\theta')=\frac{1}{2}|\phi_1-\phi_2|=\arcsin\frac{2}{M\cos\theta'}$

For conventional UPA using KPC-based HBF, by fixing $\phi$ as $\phi=\arcsin\frac{\sin\phi'}{\cos\theta'}$ in \eqref{eq:KPC response},
all the possible adjacent nulls regarding $\theta$ need to satisfy either one of the following equations
\begin{equation}
    \cos\theta\frac{\sin\phi'}{\cos\theta'}-\sin\phi'=\pm\frac{2}{M} {\text{ or}}\enspace
    \sin\theta-\sin\theta'=\pm\frac{2}{M}.
\end{equation}

By solving for $\theta$, we have
\begin{equation}
\begin{aligned}
   & \theta=(-1)^n\arcsin(\sin\theta'\pm\frac{2}{M})+n\pi,n\in\mathbb{Z}, \\
   {\text{or}}\enspace &\theta=\pm\arccos\Big(\cos\theta'\Big(1\pm\frac{2}{M\sin\phi'}\Big)\Big)+2m\pi,m\in\mathbb{Z}.
\end{aligned}
\end{equation}

Then we find the sets of nulls that are bigger and smaller than $\theta'$, denoted by $\Theta^+= \{\theta|\theta>\theta'\},\Theta^-=\{\theta|\theta<\theta'\}$. According to \eqref{eq:ltheta} and \eqref{eq:angular resolution vertical}, we have
\begin{equation}
    \theta_1=\min\Theta^+,\theta_2=\max\Theta^-.
\end{equation}

\begin{equation}
\begin{aligned}
    &\gamma_v^{\mathrm{UPA}}(\phi',\theta')=\frac{1}{2}\Delta_\theta(\phi',\theta')=\frac{1}{2}|\theta_1-\theta_2|\\
    &=\frac{1}{2}\big(\min\Theta^+-\max\Theta^-\big)\\
    &=\left\{\begin{matrix}
  E-F, &\mathrm{if}\ X\le1,Y>1,|Z|\ge1 \\
  G-F, &\mathrm{if}\ X>1,Y>1,|Z|\ge1 \\
  A-B,&\mathrm{if}\ X\le1,Y>1,|Z|\le1 \\
  A-C,&\mathrm{if}\ X\le1,Y\le1,|Z|\le1 \\
  D-B,&\mathrm{if}\ X>1,Y>1,|Z|\le1\\
  D-C,&\mathrm{if}\ X>1,Y\le1,|Z|\le1
\end{matrix}\right.,
\end{aligned}
\end{equation}
where 
$A=\min\big\{\arcsin(\sin|\theta'|+\frac{2}{M}),\arccos(\cos|\theta'|(1-\frac{2}{M\sin|\phi'|}))\big\}$, 
$B=\max\big\{ \arcsin(\sin|\theta'|-\frac{2}{M}) , -\arccos(\cos|\theta'|(1-\frac{2}{M\sin|\phi'|}))\big\}$,
$C=\max\big\{ \arcsin(\sin|\theta'|-\frac{2}{M}) , \arccos(\cos|\theta'|(1+\frac{2}{M\sin|\phi'|}))\big\}$,
$D=\max\big\{ \pi-\arcsin(\sin|\theta'|-\frac{2}{M}) , \arccos(\cos|\theta'|(1-\frac{2}{M\sin|\phi'|}))\big\}$,
$E=\arcsin(\sin|\theta'|+\frac{2}{M})$,
$F=\arcsin(\sin|\theta'|-\frac{2}{M})$,
$G=\pi-\arcsin(\sin|\theta'|-\frac{2}{M})$
$X=\sin|\theta'|+\frac{2}{M}$, 
$Y=\cos|\theta'|(1+\frac{2}{M\sin|\phi'|})$, 
$Z=\cos|\theta'|(1-\frac{2}{M\sin|\phi'|})$.

Similarly, the possible adjacent nulls regarding $\phi$ need to satisfy the following equation
\begin{equation}
    \cos\theta'\sin\phi-\sin\phi'=\pm2/M.
\end{equation}

By solving for $\phi$, we have
\begin{equation}
    \phi=(-1)^n\arcsin\Big(\frac{\sin\phi'\pm2/M}{\cos\theta'}\Big)+n\pi,n\in\mathbb{Z}.
\end{equation}

Then we find the sets of nulls that are bigger and smaller than $\phi'$, denoted by $\Phi^+=\{\phi|\phi>\phi'\}$ and $\Phi^-=\{\phi|\phi<\phi'\}$, respectively. According to \eqref{eq:lvarphi} and \eqref{eq:angular resolution vertical}, we have
\begin{equation}
    \phi_1=\min\Phi^+,\phi_2=\max\Phi^-.
\end{equation}

\begin{equation}
\begin{aligned}
    \gamma_h^{\mathrm{UPA}}(\phi',\theta')&=\frac{1}{2}\Delta_\phi(\phi',\theta')=\frac{1}{2}|\phi_1-\phi_2|\\
    &=\frac{1}{2}\big(\min\Phi^+-\max\Phi^-\big).
\end{aligned}
\end{equation}

The complete equation is given in \eqref{eq:long}



\begin{figure*}[b]
\begin{equation}\label{eq:long}
\begin{aligned}
   \gamma_h^{\mathrm{UPA}}(\phi',\theta')
    =\left\{\begin{matrix}
  \arcsin\Big(\frac{\sin\phi'+2/M}{\cos\theta'}\Big)-\arcsin\Big(\frac{\sin\phi'-2/M}{\cos\theta'}\Big)& \mathrm{if}\ |\frac{\sin\phi'+2/M}{\cos\theta'}|<1,\ |\frac{\sin\phi'-2/M}{\cos\theta'}|<1\\
  \pi-2\arcsin\Big(\frac{\sin\phi'-2/M}{\cos\theta'}\Big)&\mathrm{if}\ |\frac{\sin\phi'+2/M}{\cos\theta'}|>1,\ |\frac{\sin\phi'-2/M}{\cos\theta'}|<1 \\
  \pi+\arcsin\Big(\frac{\sin\phi'+2/M}{\cos\theta'}\Big)+\arcsin\Big(\frac{\sin\phi'+2/M}{\cos\theta'}\Big)&\mathrm{if}\ |\frac{\sin\phi'+2/M}{\cos\theta'}|<1,\ |\frac{\sin\phi'-2/M}{\cos\theta'}|>1 \\
  \mathrm{NaN}&\mathrm{if}\ |\frac{\sin\phi'+2/M}{\cos\theta'}|>1,\ |\frac{\sin\phi'-2/M}{\cos\theta'}|>1
\end{matrix}\right.
\end{aligned}.
\end{equation}
\end{figure*}

\bibliographystyle{IEEEtran}
\bibliography{IEEEabrv,0reference}


 




\vfill

\end{document}